\documentclass[lettersize,journal]{IEEEtran}
\usepackage{amsmath,amsfonts}
\usepackage{array}
\usepackage[caption=false,font=normalsize,labelfont=sf,textfont=sf]{subfig}
\usepackage{textcomp}
\usepackage{stfloats}
\usepackage{url}
\usepackage{verbatim}
\usepackage{graphicx}
\usepackage{cite}
\usepackage[linesnumbered,ruled,vlined]{algorithm2e}
\usepackage{amsmath}
\usepackage{amssymb}
\usepackage{mathtools}
\usepackage{amsthm}
\usepackage{booktabs}
\usepackage{multirow}
\usepackage{pstcol}
\usepackage{pifont}

\usepackage{siunitx}
\usepackage{booktabs}
\usepackage{tabularx}

\theoremstyle{plain}
\newtheorem{theorem}{Theorem}

\theoremstyle{definition}

\theoremstyle{remark}

\begin{document}

\title{Data-Locality-Aware Task Assignment and Scheduling for Distributed Job Executions}

\author{Hailiang~Zhao,
        Xueyan~Tang,~\IEEEmembership{Senior~Member,~IEEE,}
        Peng~Chen,
        Jianwei~Yin,
        and~Shuiguang~Deng,~\IEEEmembership{Senior~Member,~IEEE}
        % <-this % stops a space
    \thanks{H. Zhao is with the School of Software Technology, 
    Zhejiang University, 315048 Ningbo, China. E-mail: hliangzhao@zju.edu.cn.}% <-this % stops a space
    \thanks{X. Tang is with the College of Computing and Data Science, 
    Nanyang Technological University, 639798 Singapore. E-mail: asxytang@ntu.edu.sg.}% <-this % stops a space
    \thanks{P. Chen, J. Yin and S. Deng are with the College of Computer Science and Technology, 
    Zhejiang University, 310027 Hangzhou, China. E-mails: naturechenpeng@gmail.com, \{zjuyjw, dengsg\}@zju.edu.cn.}% <-this % stops a space
    \thanks{Xueyan Tang is the corresponding author.}
}

% {Zhao \MakeLowercase{\textit{et al.}}: Data-Locality Aware Task Assignment and Scheduling for Distributed Job Executions}


\maketitle

\begin{abstract}
    This paper addresses the data-locality-aware task assignment and scheduling problem for distributed job executions. Our goal is to minimize job completion times without prior knowledge of future job arrivals. We propose an Optimal Balanced Task Assignment algorithm (OBTA), which achieves minimal job completion times while significantly reducing computational overhead through efficient narrowing of the solution search space. To balance performance and efficiency, we extend the approximate Water-Filling (WF) algorithm, providing a rigorous proof that its approximation factor equals the number of task groups in a job. We also introduce a novel heuristic, Replica-Deletion (RD), which outperforms WF by leveraging global optimization techniques. To further enhance scheduling efficiency, we incorporate job ordering strategies based on a shortest-estimated-time-first policy, reducing average job completion times across workloads. Extensive trace-driven evaluations validate the effectiveness and scalability of the proposed algorithms.
\end{abstract}

\begin{IEEEkeywords}
    task assignment, job scheduling, distributed job execution, approximate analysis, trace-driven evaluation.
\end{IEEEkeywords}

\section{Introduction}\label{s1}
% \IEEEPARstart{S}{cheduling} plays a critical role in the performance of big data analysis and high-performance computing, where jobs are executed across multiple distributed servers \cite{reuther2018scalable,wang2020survey}. These jobs typically consist of numerous tasks, each requiring access to distinct data chunks that might be replicated across different locations. Data locality, which involves assigning tasks to servers holding the required data chunks, is crucial in this context \cite{guo2012investigation}. Designing online scheduling algorithms that uphold data locality in distributed job executions, without prior knowledge of job arrivals, involves two key aspects: \textit{task assignment} and \textit{job reordering}. Task assignment involves allocating tasks within a job to appropriate servers, conceptualized as a semi-matching problem in a bipartite graph where tasks and servers are distinct sets of nodes. Job reordering addresses the sequence in which the arrived but unfinished job tasks are executed, facilitating optimized task assignments and reducing overall job completion times.
\IEEEPARstart{S}{cheduling} is a fundamental component of modern service-oriented computing, enabling distributed systems to efficiently execute diverse workloads and deliver scalable services \cite{reuther2018scalable,wang2020survey,10213224,9925644}. In cloud-based environments, jobs typically comprise multiple tasks that require access to specific data chunks, which are often replicated across geographically distributed servers. Efficient data-locality-aware scheduling, which assigns tasks to servers hosting the required data chunks, is critical for reducing data transfer overhead, enhancing task execution efficiency, and ensuring compliance with stringent service-level agreements (SLAs) \cite{guo2012investigation}. The design of online scheduling algorithms for such multi-task jobs introduces two central challenges: \textit{task assignment} and \textit{job ordering}. Task assignment entails mapping tasks to appropriate servers based on data locality, often modeled as a semi-matching problem on a bipartite graph with tasks and servers as nodes. Job ordering involves determining the execution sequence of pending tasks to optimize resource utilization and minimize job completion times, a key requirement in dynamic and resource-intensive service environments.

Existing research has made progress on these challenges, yet often under assumptions that limit applicability in real-world service systems \cite{hung2015scheduling,10.1145/3337821.3337843,beaumont2020performance,9826037,10213224,9925644}. For instance, Hung \textit{et al.} \cite{hung2015scheduling} proposed a workload-aware greedy scheduling strategy (SWAG) that prioritizes jobs based on estimated completion times. While effective in specific scenarios, this approach neglects data replication, a prevalent feature in distributed service systems. Similarly, Guan \textit{et al.} \cite{10.1145/3337821.3337843} introduced an offline fairness-aware task assignment algorithm to ensure max-min fairness across servers, but their method lacks the adaptability needed for online and dynamic service platforms. Beaumont \textit{et al.} \cite{beaumont2020performance} explored scheduling with data replication but constrained their analysis to single-job scenarios and homogeneous environments, which oversimplifies the heterogeneity and complexity inherent in service-oriented infrastructures. Guan \textit{et al.} \cite{9826037} tackled an online scheduling problem for distributed job executions, leveraging maximum flow-based techniques for task assignment. However, their work offered limited performance analysis and did not thoroughly examine job ordering strategies essential for dynamic service environments.

In this paper, we focus on developing data-locality-aware task assignment and scheduling algorithms with theoretical guarantees tailored for online distributed service environments. Our approach addresses two fundamental scenarios: (i) jobs are executed in a FIFO (first-in-first-out) manner, and (ii) jobs can be prioritized and reordered to optimize performance. By addressing these scenarios, we aim to improve the efficiency of task scheduling in service-oriented systems, tackling the dual challenges of data locality and dynamic workload management. The main contributions of this work are summarized as follows:
\begin{itemize}
    \item We formulate the task assignment problem as a bipartite graph-based non-linear program and propose an Optimal Balanced Task Assignment algorithm (OBTA), which significantly reduces the search space and provides efficient solutions for heterogeneous service environments.
    
    \item We extend the approximate water-filling (WF) algorithm, originally proposed for homogeneous settings in \cite{9826037}, to heterogeneous settings. We provide a nontrivial proof that WF achieves a $K$-approximation, where $K$ represents the number of task groups in a job.

    \item We propose a novel heuristic, replica-deletion (RD), for task assignment. RD generally produces better task assignments than WF.
    % RD assigns tasks to servers with a computational complexity of $ \mathcal{O}(M^2 \cdot n \log n)$, where $M$ is the number of servers and $n$ is the number of tasks in a job.
    
    \item We extend WF to support job reordering and introduce an accelerated order-conscious WF (OCWF-ACC) algorithm, which incorporates an early-exit technique to improve computational efficiency.
    
    \item We validate the proposed algorithms through extensive trace-driven evaluations. OBTA demonstrates significant reductions in average job completion times, while WF provides a low-overhead yet effective approximation. RD enhances overall performance further. Moreover, OCWF-ACC achieves substantial reductions in computation overhead compared to job-reordered task assignment algorithms without early-exit techniques.
\end{itemize}

The remainder of this paper is organized as follows: Sec. \ref{s2} defines the problem formulation. Sec. \ref{s3} presents task assignment algorithms designed for FIFO scheduling. Sec. \ref{s4} extends these algorithms to support job reordering integration. Sec. \ref{s5} provides experimental results based on real-world job traces, demonstrating the effectiveness of the proposed methods. Sec. \ref{s6} reviews related work, highlighting the context and contributions of this study. Finally, Sec. \ref{s7} concludes the paper and discusses potential future research directions.

\section{Problem Formulation}\label{s2}

We consider a distributed computing system comprising $M$ servers, indexed by $m$ for $m \in \mathcal{M} := \{1, ..., M \}$. The system stores a collection of data items, such as key-value pairs, objects, or files, partitioned into equally-sized data chunks. Each server holds a subset of these data chunks, and each chunk can be replicated across multiple servers. For each server $m \in \mathcal{M}$, let $\mathcal{D}_m$ denote the set of data chunks stored on it, where $\mathcal{D}_m \cap \mathcal{D}_{m’}$ can be non-empty for $m \neq m’$. \textcolor{black}{Key notations and acronyms used in this paper are summarized in Table \ref{tab_notation} and Table \ref{tab_acronym}, respectively.}
% The distribution of data chunks across servers is unknown, but is assumed to be static and given at the outset.

\begin{table}[ht]
    \color{black}
    \centering
    \caption{Summary of key notations.}
    \label{tab_notation}
    \begin{tabular}{cl}
    \toprule
    \textbf{Notation} & \textbf{Description} \\
    \midrule
    $\mathcal{M}$ & The set of servers \\
    $m$ & Index of servers in $\mathcal{M}$ \\
    $\mathcal{D}_m$ & The set of data chunks in server $m \in \mathcal{M}$ \\
    $c$ & Index of arriving jobs \\
    $\mathcal{T}_c$ & The set of tasks in job $c$ \\
    $r$ & Index of tasks \\
    $\mathcal{S}^r$ & The available servers (due to data locality) of task $r$ \\
    $\mu_m^c$ & Processing capacity of server $m$ for job $c$ \\
    $\Phi_c$ & The estimated completion time of job $c$ \\
    $\Phi_c^+ (\Phi_c^-)$ & The upper (lower) bound of $\Phi_c$ derived by OBTA \\
    $b_m^c$ & The busy time of server $m$ at the arrival of job $c$ \\
    $K_c$ & The number of task groups in job $c$ \\
    $k$ & Index of task groups \\
    $\mathcal{T}_c^k$ & The set of tasks in the $k$-th group of job $c$ \\
    $\mathcal{S}_c^k$ & The set of available servers of the task group $\mathcal{T}_c^k$ \\
    $\textsc{ALG}$ & An algorithm \textsc{ALG} (and its objective value) \\
    $\alpha_{\textsc{ALG}}$ & The approximation ratio of the algorithm \textsc{ALG} \\
    $I$ & A job arrival instance \\
    $\mathcal{O}_c$ & The set of outstanding jobs at the arrival of job $c$ \\
    \bottomrule
    \end{tabular}
\end{table}

\begin{table}[ht]
    \color{black}
    \centering
    \caption{Summary of acronyms of related algorithms. OBTA-N and OBTA-P are introduced in Sec. \ref{s6}.}
    \label{tab_acronym}
    \begin{tabular}{cl}
    \toprule
    \textbf{Acronym} & \textbf{Description} \\
    \midrule
    SWAG & workload-aware greedy scheduling strategy \\
    OBTA & optimal balanced task assignment \\
    NLIP & non-linear integer programming \\
    WF & water-filling \\
    RD & replica deletion \\
    OCWF & order-conscious water-filling \\
    OCWF-ACC & accelerated order-conscious water-filling \\
    LIP & linearized integer programming \\
    OBTA-N & OBTA with only solution space narrowing \\
    OBTA-P & OBTA with only piecewise linearization \\
    \bottomrule
    \end{tabular}
\end{table}

% Jobs arrive dynamically in the system, one at a time, for execution. Each job consists of multiple independent tasks, where each task processes a specific data partition. A task requires exactly one data chunk as input, and all tasks within a job are assumed to have identical computational demands. We denote by $ c \in \mathbb{N}^+ $ the $ c $-th job (also referred to as job $ c $) in chronological order. Further, we use $ \mathcal{T}_c $ to represent the set of tasks in job $ c $, and $ d_r $ to represent the data chunk demanded by a task $ r \in \mathcal{T}_c $. For each task $ r \in \mathcal{T}_c $, we define
% \begin{equation}
%     \mathcal{S}^r := \big\{ m \in \mathcal{M} \mid d_r \in \mathcal{D}_m \big\}
% \end{equation}
% as the set of its \textit{available servers}. To ensure data locality, each task must be assigned to one of its available servers for execution. Each server $m$ maintains a queue $q_m$ to store its outstanding tasks, with no predefined limit on the queue size. For analysis purposes, we model the system’s time as divided into discrete, identical time slots.
Jobs arrive dynamically, one at a time, for execution. Each job consists of multiple independent tasks, where each task processes a specific data partition. A task requires exactly one data chunk as input, and all tasks within a job are assumed to have identical computational demands. Let $c \in \mathbb{N}^+$ denote the $c$-th job (also referred to as job $c$) in chronological order. Further, let $\mathcal{T}_c$ represent the set of tasks in job $c$, and let $d_r$ denote the data chunk required by a task $r \in \mathcal{T}_c$. For each task $r \in \mathcal{T}_c$, the set of available servers capable of processing the task is defined as:
\begin{equation}
    \mathcal{S}^r := \big\{ m \in \mathcal{M} \mid d_r \in \mathcal{D}_m \big\}.
\end{equation}
To maintain data locality, each task must be assigned to one of its available servers for execution. Each server $m$ maintains a queue $q_m$ to store outstanding tasks, with no predefined limit on queue size. For analysis purposes, the system’s time is divided into discrete, identical time slots.

% Each server $m$ has a task processing capacity, characterized by $\mu_m^c$, which represents the number of tasks from job $c$ that server $m$ can process within a single time slot. We define the initial \textit{busy time} of each server $ m $ just before job $c$'s arrival, denoted by $ b_m^c $, as the number of time slots required to process all tasks currently queued at server $m$. The initial busy time $ b_m^c $ can be \textit{estimated} using the following equation:
% \begin{equation}
%     b_m^c = \sum_{h=1}^{c-1} \left\lceil \frac{o_m^h}{\mu_m^h} \right\rceil,
%     \label{bklg}
% \end{equation}
% where $ o_m^h \leq |\mathcal{T}_h| $ is the number of tasks from job $h$ that are still pending in $q_m$ when job $c$ arrives. 
Each server $m$ has a task processing capacity characterized by $\mu_m^c$, representing the number of tasks from job $c$ that server $m$ can process within a single time slot. The initial busy time of each server $m$ just before job $c$’s arrival, denoted by $b_m^c$, is the number of time slots required to process all tasks currently queued at server $m$. This busy time can be estimated using:
\begin{equation}
    b_m^c = \sum_{h=1}^{c-1} \left\lceil \frac{o_m^h}{\mu_m^h} \right\rceil,
    \label{bklg}
\end{equation}
where $o_m^h \leq |\mathcal{T}_h|$ is the number of tasks from job $h$ that remain pending in $q_m$ when job $c$ arrives.
%Additionally, we denote the queue size of server $ m $ at time $ t_c $ as $ |q_m| := \sum_{h < c} o_m^h $.

The objective is to design a scheduling strategy that assigns tasks to servers and determines execution priorities to minimize the completion time of each job. The completion time of a job is defined as the duration between its arrival and the completion of its last processed task, equivalent to minimizing the finish time of the last task in the job. We investigate the problem under two distinct scenarios:
\begin{itemize}
    \item \textbf{FIFO queues.} Outstanding tasks are processed in the order of their arrival (first-in-first-out). The challenge is to assign tasks of each arriving job to available servers while balancing server workloads to minimize job completion times.
    \item \textbf{Prioritized reordering.} Outstanding tasks can be dynamically reordered in each server’s queue. This scenario focuses on prioritizing tasks to further reduce job completion times, introducing additional complexity as reordering decisions influence overall system performance.
\end{itemize}

We do not impose any specific assumptions about the timing or pattern of job arrivals.

\section{Algorithms for FIFO Queues}\label{s3}

\subsection{Matching-based Non-Linear Programming}\label{s31}

\subsubsection{Programming in a Bipartite Graph}\label{s311}
When a new job $c$ arrives, the task assignment problem involves allocating all tasks in $\mathcal{T}_c$ to appropriate servers while ensuring data locality and minimizing the job completion time.

To facilitate this, we divide the tasks into task groups, where each task group contains tasks sharing the same set of available servers. Let $K_c$ denote the number of task groups for job $c$. Further, for each $ k \in \mathcal{K}_c := \{ 1, \ldots, K_c \} $, we denote by $ \mathcal{T}_c^k $ and $ \mathcal{S}_c^k $ the set of tasks in the $ k $-th group (also referred to as group $ k $) and the set of available servers for tasks in the $ k $-th group, respectively. By definition, we have
\begin{equation}
    \mathcal{T}_c^k = \left\{ r \in \mathcal{T}_c \mid \mathcal{S}^r = \mathcal{S}_c^k \right\}.
\end{equation}

The task assignment problem can be modeled as a matching problem on a bipartite graph. In this graph, the left-side nodes are the $ K_c $ task groups, while the right-side nodes are the $ M $ servers. For each $ k \in \mathcal{K}_c $, an edge $(k, m)$ exists if $ m \in \mathcal{S}_c^k $. We can formulate the assignment problem as a non-linear integer program as follows:
\begin{align}
    &\mathcal{P}: \qquad \qquad \qquad \min_{\Phi_c, n_m^k \in \mathbb{N}^+} \Phi_c \nonumber \\
    \text{s.t.} &\left\{ 
        \begin{array}{ll}
            \sum_{k \in \mathcal{K}_c} n_m^k \leq \max \left\{ \Phi_c - b_m^c, 0 \right\} & \forall m \in \bigcup_{k \in \mathcal{K}_c} \mathcal{S}_c^k, \\
            \sum_{m \in \mathcal{S}_c^k} n_m^k \cdot \mu_m^c \geq |\mathcal{T}_c^k| & \forall k \in \mathcal{K}_c.
        \end{array}
    \right. \label{cons}
\end{align}
Recall that $ b_m^c $ is the (initial) estimated busy time of server $ m $ at the time when job $c$ arrives, calculated by \eqref{bklg}. In $ \mathcal{P} $, $ n_m^k $ is the number of time slots required to process the tasks in group $ \mathcal{T}_c^k $ at server $ m $, given the profiled processing capacity $ \mu_m^c $. Each $ n_m^k $ is associated with an edge $(k,m)$ in the bipartite graph. $ \Phi_c $ is the estimated completion time of job $ c $. \textcolor{black}{The first constraint set in \eqref{cons} means that any available server whose initial estimated busy time is larger than $ \Phi_c $ should not participate in the assignment of tasks in job $c$.} This constraint ensures balance among the participating servers. \textcolor{black}{The second constraint set in \eqref{cons} ensures that all the tasks in job $c$ can be processed.} The program $\mathcal{P}$ can be solved optimally using commercial solvers such as CPLEX.\footnote{\url{https://www.ibm.com/products/ilog-cplex-optimization-studio}}

\subsubsection{Narrowing the Search of $\Phi_c$}\label{s312}
Commercial solvers can be computationally expensive when solving non-linear integer programs, especially in large-scale scenarios. To mitigate this overhead, we narrow the search space of $\Phi_c$ by leveraging the inherent characteristics of the problem.

After the arrival of job $c$, the total number of time slots required to process both the backlogged tasks and the newly assigned tasks from job $c$ is bounded above by:
\begin{equation}
    \Phi_c^{+} := \max_{m \in \bigcup_{k \in \mathcal{K}_c} \mathcal{S}_c^k} \left\{ \left\lceil \frac{\sum_{k \in \mathcal{K}_c: m \in \mathcal{S}_c^k} |\mathcal{T}_c^k|}{\mu_m^c} \right\rceil + b_m^c \right\}.
\end{equation}
This upper bound, $\Phi_c^{+}$, is determined under the assumption that all tasks in $\mathcal{T}_c$ are assigned to a single available server $m$. The corresponding lower bound of $\Phi_c$ is defined as:
\begin{equation}
    \Phi_c^{-} := \max_{k \in \mathcal{K}_c} x_k,
    \label{xk}
\end{equation}
where $x_k$ is the minimum integer satisfying:
\begin{equation}
    \sum_{m \in \mathcal{S}_c^k} \left( \max \left\{ x_k - b_m^c, 0 \right\} \right) \cdot \mu_m^c \geq |\mathcal{T}_c^k|.
    \label{xk2}
\end{equation}
Here, $x_k$ represents the minimal number of time slots required to process all tasks in $\mathcal{T}_c^k$, assuming that task group $k$ is the only group being processed for job $c$. By combining these bounds, we reduce the search space of $\Phi_c$ from $\mathbb{N}^+$ to the integer interval $\left[ \Phi_c^-, \Phi_c^+ \right]$. This narrowing significantly reduces the computational complexity of solving the task assignment problem, making it more tractable for large-scale systems.

\subsubsection{The Matching-based Algorithm}\label{s313}

\begin{figure}[ht]
  \centering
  \includegraphics[width=3in]{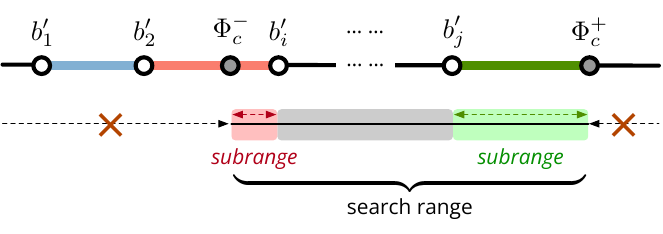}
  \caption{Divide $\big[\Phi_c^-, \Phi_c^+ \big]$ into disjoint subranges.}
  \label{dividing}
\end{figure}

% By introducing \( \Phi_c^{-} \) and \( \Phi_c^{+} \), the search space of \( \Phi_c \) is greatly narrowed. However, \( \mathcal{P} \) is still non-linear since the first constraint in \eqref{cons} is piecewise. To solve \( \mathcal{P} \) efficiently, we divide \( \left[ \Phi_c^-, \Phi_c^+ \right] \) into several disjoint sub-intervals based on the initial estimated busy times of available servers. The division is illustrated in Fig. \ref{dividing}. Specifically, we sort the initial estimated busy time \( b_m^c \) of each server \( m \in \bigcup_{k \in \mathcal{K}_c} \mathcal{S}_c^k \) in ascending order, and we denote the sorted values by \( b_1', b_2', \ldots \). We then denote by \( i \) (\( j \)) the minimal (maximal) index of these values such that \( b_i' \geq \Phi_c^- \) (\( b_j' \leq \Phi_c^+ \)). Then, we can solve \( \mathcal{P} \) with the search space of \( \Phi_c \) falling into each sub-interval \( [\Phi_c^-, b_i'), [b_{i}', b_{i+1}'), \ldots, [b_j', \Phi_c^+] \) in turn. Each of these problems is transformed into a linear integer program since the first constraint of \eqref{cons} is no longer piecewise. The sub-intervals are checked one by one. If \( \mathcal{P} \) is solvable within the current sub-interval, then its solution is the optimal one since the remaining sub-intervals cannot contain a smaller \( \Phi_c \).
By introducing $\Phi_c^{-}$ and $\Phi_c^{+}$, the search space of $\Phi_c$ is significantly narrowed. However, the problem $\mathcal{P}$ remains non-linear due to the piecewise nature of the first constraint in \eqref{cons}. To address this, we divide the interval $\left[ \Phi_c^-, \Phi_c^+ \right]$ into multiple disjoint sub-intervals based on the initial estimated busy times of the available servers. This division is illustrated in Fig. \ref{dividing}. Specifically, the initial estimated busy times $b_m^c$ of all servers $m \in \bigcup_{k \in \mathcal{K}_c} \mathcal{S}_c^k$ are sorted in ascending order, and the sorted values are denoted as $b_1’, b_2’, \ldots$. We then identify indices $i$ and $j$, where $i$ is the smallest index such that $b_i’ \geq \Phi_c^-$, and $j$ is the largest index such that $b_j’ \leq \Phi_c^+$. With this setup, the search space of $\Phi_c$ is divided into sub-intervals: $\left[ \Phi_c^-, b_i’ \right)$, $\left[ b_i’, b_{i+1}’ \right)$, ..., $\left[ b_j’, \Phi_c^+ \right]$. Within each sub-interval, the piecewise nature of the first constraint in \eqref{cons} is resolved, transforming $\mathcal{P}$ into a linear integer program. These sub-problems are then solved sequentially. If $\mathcal{P}$ is solvable within a particular sub-interval, the corresponding solution is guaranteed to be optimal, as the remaining sub-intervals cannot yield a smaller $\Phi_c$.

Using this approach, we propose the Optimal Balanced Task Assignment (OBTA) algorithm. When the profiled values $\{ \mu_m^c \}_{m,c}$ are accurate, OBTA produces optimal solutions by computing the exact values of $n_m^k$ for task assignment. The pseudocode for OBTA is provided in Algorithm \ref{algo-OBTA}.

\begin{algorithm}[t]
  \caption{OBTA}
  \label{algo-OBTA}
  {\color{black}\KwIn{Online arriving jobs $1, 2, ..., c, ...$, server capacities $\{ \mu_m^c \}_{m,c}$}}
  {\color{black}\KwOut{Assignment solution for each task of each arriving job}}
  \While{a new job $c$ arrives}
  {
    \ForEach{$m \in \bigcup_{k \in \mathcal{K}_c} \mathcal{S}_c^k$} %{\normalfont \textbf{in parallel}}}
    {
      Estimate $b_m^c$ using \eqref{bklg}
    }
    Solve $\mathcal{P}$ as described in Sec. \ref{s313}, obtaining $\{ n_m^k \}_{m,k}$\\
    \ForEach{$k \in \mathcal{K}_c$}
    {
      % Sort the servers in $\mathcal{S}_c^k$ in ascending order of their estimated busy times, i.e., $b_m^c$ \textcolor{black}{is this the initial busy time or the busy time after assiging groups 1, 2, ..., k-1? seems you don't need sorting here at all}\\
      \ForEach{$m \in \mathcal{S}_c^k$}
      {
        \uIf{server $m$ is not the last server in $\mathcal{S}_c^k$}
        {
          Assign $n_m^k \cdot \mu_m^c$ group-$k$ tasks to server $m$
        }
        \Else{
          Assign all remaining group-$k$ tasks to server $m$
        }
      }
    }
  }
\end{algorithm}

\subsection{Water-Filling}\label{s32}
Solving a group of integer programs can still be time-consuming for large-scale services environments. To improve efficiency, we propose an alternative task assignment algorithm called Water-Filling (WF), which is extended from \cite{9826037}. WF assigns tasks to jobs sequentially, one group at a time, using an approximate approach to balance computational efficiency with performance.

\subsubsection{Algorithm Design}\label{sec-wf-design}
WF dynamically estimates server workloads as tasks are assigned. Let $b_m^c(k)$ denote the estimated busy time of server $m$ after assigning tasks from group $k$. Initially, for each server $m \in \mathcal{M}$, the busy time is:
\begin{equation}
    b_m^c(0) := b_m^c,
\end{equation}
and it increases monotonically as task groups are assigned:
\begin{equation}
    b_m^c(0) \leq b_m^c(1) \leq \cdots \leq b_m^c \left(K_c \right).
\end{equation}
For each task group $k$, WF determines the minimum time $\xi_k$ required to process all tasks in the group across its available servers. Specifically, $\xi_k$ is the smallest integer satisfying:
\begin{equation}
    \sum_{m \in \mathcal{S}_c^k} \left( \max \left\{ \xi_k - b_m^c(k-1), 0 \right\} \right) \cdot \mu_m^c \geq |\mathcal{T}_c^k|,
    \label{xik}
\end{equation}
where $\mathcal{S}_c^k$ is the set of servers available for group $k$, and $b_m^c(k-1)$ represents the server’s busy time after the previous task group assignment. The calculation of $\xi_k$ is analogous to $x_k$ in \eqref{xk2}, representing the estimated time needed to process all tasks in the group. A server $m \in \mathcal{S}_c^k$ is considered a participating server if $\xi_k - b_m^c(k-1) > 0$. We then allocate $ (\xi_k - b_m^c(k-1)) \cdot \mu_m^c $ tasks in group $k$ to each participating server $ m $ (or all the remaining tasks in group $ k $ if server $ m $ is the last participating server). After the assignment, the busy time of each server is updated as:
\begin{equation}
    b_m^c(k) = \max \left\{ b_m^c(k-1), \xi_k \right\}.
    \label{bk_update}
\end{equation}
This process, repeated for each task group, simulates a water-filling mechanism where the participating servers’ workloads are gradually equalized. Fig. \ref{greedy} illustrates the water-filling concept for task assignment.

\begin{figure}[t]
  \centering
  \includegraphics[width=2.25in]{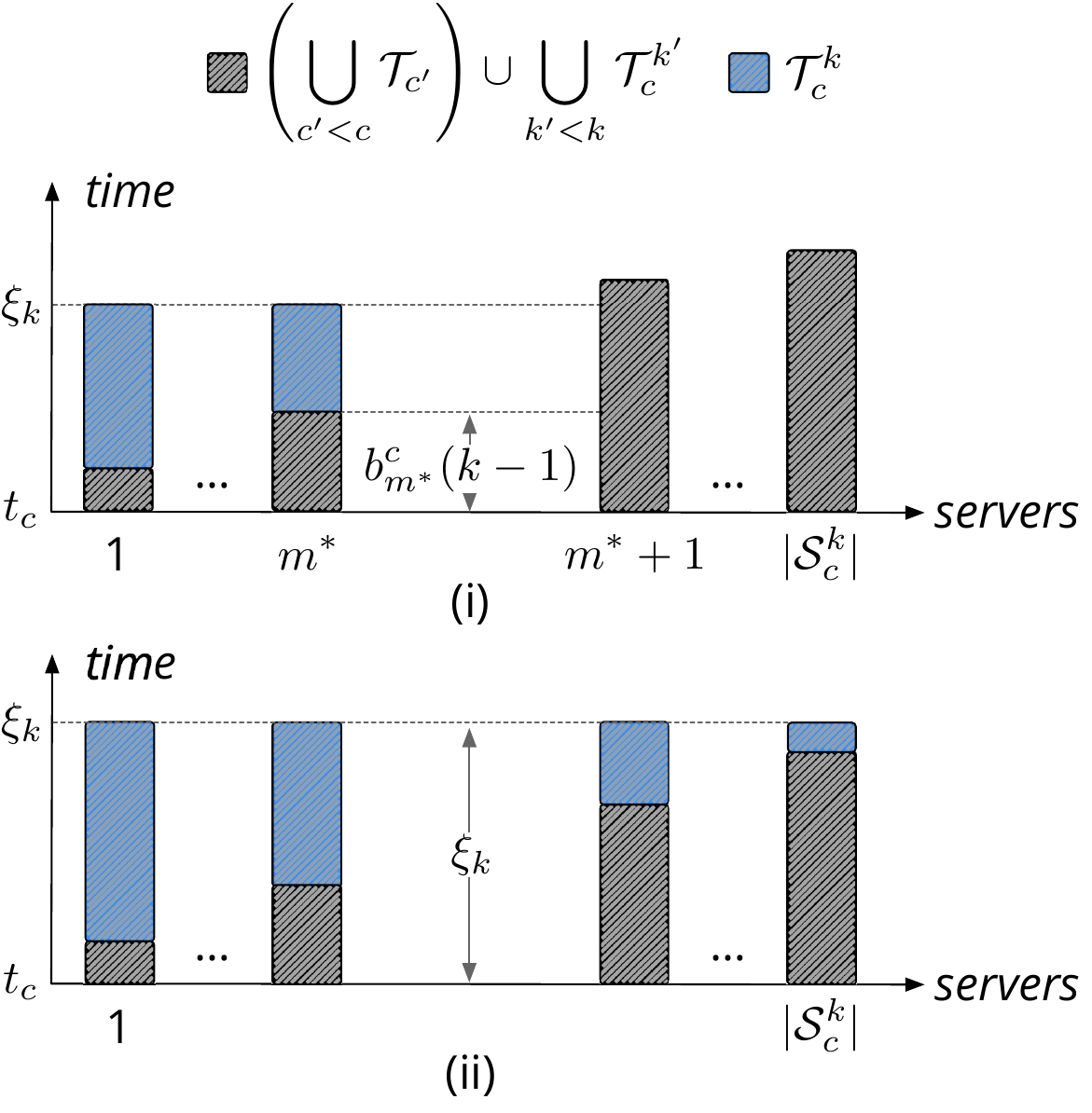}
  \caption{Assign tasks in $\mathcal{T}_c^k$ to their available servers in a water-filling manner. In (i), only a subset of $\mathcal{S}_c^k$, including servers $1, ..., m^*$, participate in the assignment of $\mathcal{T}_c^k$; In (ii), every server in $\mathcal{S}_c^k$ participates in the assignment of $\mathcal{T}_c^k$.}
  \label{greedy}
\end{figure}

The pseudocode of WF is presented in Algorithm \ref{algo-wf}. Its computational complexity is $\mathcal{O} \left( \sum_{k \in \mathcal{K}_c} |\mathcal{S}_c^k| \log (|\mathcal{S}_c^k|) \right)$, as $\xi_k$ can be calculated via binary search with $\mathcal{O}(\log(|\mathcal{S}_c^k|))$ iterations. This complexity can also be expressed as $\mathcal{O}(K_c \cdot M \cdot \log M)$.

\begin{algorithm}[t]
  \caption{WF}
  \label{algo-wf}
  {\color{black}\KwIn{Online arriving jobs $1, 2, ..., c, ...$, server capacities $\{ \mu_m^c \}_{m,c}$}}
  {\color{black}\KwOut{Assignment solution for each task of each arriving job}}
  \While{a new job $c$ arrives}
  {
    \For{each $m \in \bigcup_{k \in \mathcal{K}_c} \mathcal{S}_c^k$} %{\normalfont \textbf{in parallel}}}
    {
      Estimate $b_m^c$ by \eqref{bklg}\\
      $b_m^c(0) \leftarrow b_m^c$
    }
    \For{each $k \in \mathcal{K}_c$}
    {
      Calculate $\xi_k$ as detailed in Sec. \ref{sec-wf-design}\\
      \For{each $m \in \mathcal{S}_c^k$}
      {
        \If{$b_m^c(k-1) \geq \xi_k$}
        {
          \textbf{continue}
        }
        \uIf{server $m$ is not the last server in $\mathcal{S}_c^k$}
        {
          Assign $\left(\xi_k - b_m^c(k-1)\right) \cdot \mu_m^c$ group-$k$ tasks to server $m$
        }
        \Else{
          Assign all the remaining group-$k$ tasks to server $m$
        }
        Update $b_m^c(k)$ by \eqref{bk_update}
      }
    }
  }
\end{algorithm}

\subsubsection{Approximation Analysis}
In this section, we analyze the performance of WF. The approximation factor of WF, compared to the optimal algorithm, is defined as:
\begin{equation}
    \alpha_{\textsc{WF}} := \max_{I} \frac{\text{WF}(I)}{\text{OPT}(I)},
\end{equation}
where $ I := I \big(c, \{ b_m^c \}_m \big) $ represents the arrival instance of a new job $c$ and the initial estimated busy times of all servers prior to job $c$’s arrival. Here, $\text{OPT}(I)$ and $\text{WF}(I)$ denote the maximum estimated busy times of participating servers after assigning job $c$’s tasks using the optimal algorithm and WF, respectively.

\begin{figure}[h]
  \centering
  \includegraphics[width=2.9in]{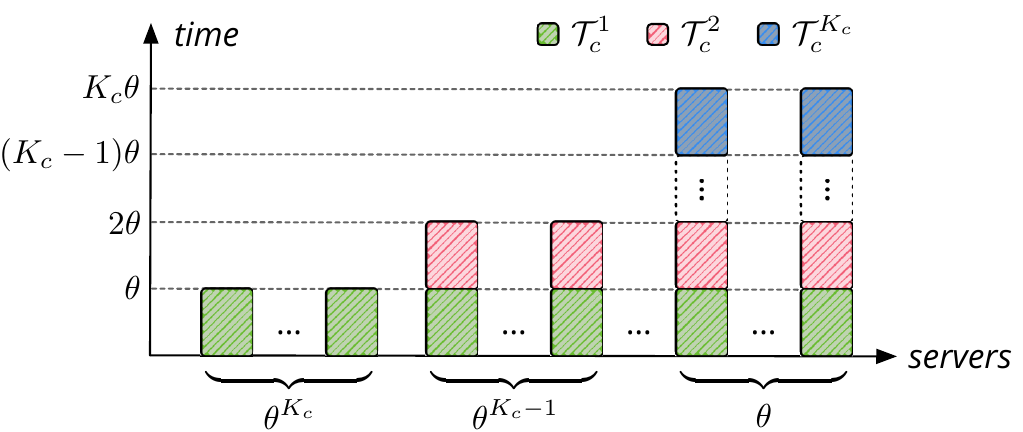}
  \caption{Visualization of the task assignment by WF for a constructed instance.}
  
  \label{wf_lower_bound_I}
\end{figure}

\begin{theorem}
    For each newly arrived job $c$, the approximation factor of WF is at least $K_c$.
    \label{theorem-wf-lb}
\end{theorem}
\begin{proof}
    The WF-to-optimal ratio is 1 under favorable conditions, such as when $\mathcal{T}_c$ consists of only one task group or when $\mathcal{S}_c^k \cap \mathcal{S}_c^{k’} = \emptyset$ for all $k \neq k’$. However, to demonstrate the worst-case ratio, we construct a specific instance $I$, illustrated in Fig. \ref{wf_lower_bound_I}, where the available servers of different task groups overlap.

    For simplicity, let $\mu_m^c = 1$ and $b_m^c = 0$ for all servers $m$. For each task group $k$, the number of its available servers is defined as:
    \begin{equation}
        |\mathcal{S}_c^k| = \sum_{k'=1}^{K_c - k + 1} x^{k'},
    \end{equation}
    where $x \geq 2$ is an integer. Furthermore, the available servers follow a hierarchical structure:
    \begin{equation}
        \mathcal{S}_c^1 \supset \mathcal{S}_c^2 \supset \cdots \supset \mathcal{S}_c^{K_c},
    \end{equation}
    meaning the servers available for a higher-indexed group are a subset of those available for a lower-indexed group. Each task group $k$ contains
    $$
    |\mathcal{T}_c^k| = x \cdot |\mathcal{S}_c^k|
    $$
    tasks. In WF, tasks in each group are assigned sequentially, and as shown in Fig. \ref{wf_lower_bound_I}, $x$ time slots are required to process all tasks in each group.
    % For simplicity, we assume $ \mu_m^c = 1 $ and $ b_m^c = 0 $ for each server $ m $. For each group $ k $, the number of its available servers is given by
    % \begin{equation}
    %     |\mathcal{S}_c^k| = \sum_{k'=1}^{K_c - k + 1} x^{k'},
    % \end{equation}
    % where $ x \geq 2 $ is an integer. Moreover,
    % \begin{equation}
    %     \mathcal{S}_c^1 \supset \mathcal{S}_c^2 \supset \cdots \supset \mathcal{S}_c^{K_c},
    % \end{equation}
    % i.e., the available servers of a higher-indexed group are a subset of those of a lower-indexed group. The number of tasks in group $ k $ is given by $ |\mathcal{T}_c^k| = x \cdot |\mathcal{S}_c^k| $. The task assignment made by WF is shown in Fig. \ref{wf_lower_bound_I}, where $ x $ time slots are required to process the tasks in each group.
    
    \begin{figure}[h]
      \centering
      \includegraphics[width=2.6in]{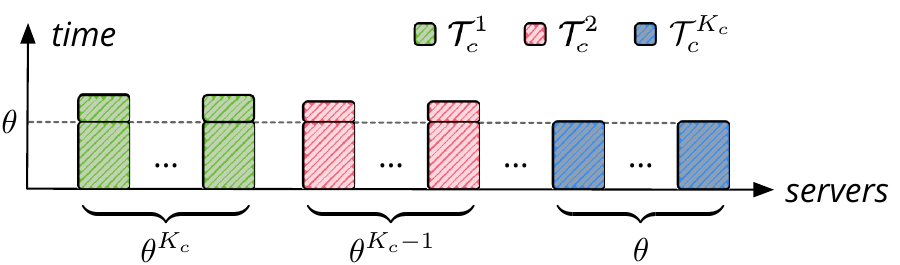}
      \caption{Visualization of the task assignment by OPT for the constructed instance.}
      \label{wf_lower_bound_I_2}
    \end{figure}
    
    In the optimal assignment (OPT), illustrated in Fig. \ref{wf_lower_bound_I_2}, tasks from group $k$ ($k < K_c$) are assigned to a subset of its servers $\mathcal{S}_c^k \setminus \mathcal{S}_c^{k+1}$. The number of time slots required to process all tasks in group $k$ is:
    \begin{equation}
        \Bigg\lceil \frac{x \cdot \sum_{i=1}^{K_c - k + 1} x^{i}}{x^{K_c - k + 1}} \Bigg\rceil = x + 2.
    \end{equation}
    For group $K_c$, the number of time slots required remains $x$. Therefore, the total completion time for job $c$ using OPT is:
    \begin{equation}
        \textsc{OPT}(I) = \max\{ x + 2, x \} = x + 2.
    \end{equation}

    The approximation factor for this instance is:
    \begin{equation}
        \frac{\textsc{WF}(I)}{\textsc{OPT}(I)} = \frac{K_c \cdot x}{x + 2},
    \end{equation}
    As $x \to \infty$, the ratio approaches $K_c$, demonstrating that the worst-case approximation factor of WF is at least $K_c$.
\end{proof}

Theorem \ref{theorem-wf-lb} establishes that $K_c$, the number of task groups in job $c$, is a lower bound for the approximation factor $\alpha_{\textsc{WF}}$ of WF. We now demonstrate that $K_c$ is also an upper bound for $\alpha_{\textsc{WF}}$, thereby tightly characterizing the approximation factor of WF.

\begin{theorem}
    For each newly arrived job $c$, the approximation factor of WF is at most $K_c$.
    \label{theorem-wf-ub}
\end{theorem}
\textcolor{black}{The proof is given in Appendix \ref{app_proof}.}

Theorems \ref{theorem-wf-lb} and \ref{theorem-wf-ub} collectively establish that WF has a tight approximation ratio of $K_c$, where $K_c$ is the number of task groups in the incoming job $c$.

\subsection{Replica-Deletion}\label{s33}
In this section, we introduce the Replica-Deletion (RD) heuristic, a task assignment strategy designed to minimize job completion time by iteratively balancing server workloads through the removal of redundant task replicas.

\subsubsection{Algorithm Design}\label{s331}
% We illustrate RD with an example, as shown in Fig. \ref{rd}. When a new job $c$ arrives, each task $r$ in $\mathcal{T}_c$ is initially replicated $|\mathcal{S}^r$ times at all of its available servers. For instance, the blue task in Fig. \ref{rd} has three available servers: servers 1, 2, and 5, so it is replicated at these three servers in Fig. \ref{rd}(i). 
We illustrate RD with an example, as shown in Fig. \ref{rd}. 

\textbf{Initialization phase.}
When a new job $c$ arrives, each task $r \in \mathcal{T}_c$ is initially replicated $|\mathcal{S}^r$ times across all its available servers. For example, in Fig. \ref{rd}(i), the blue task has three available servers (servers 1, 2, and 5), so it is replicated on these three servers.

\begin{figure}[h]
    \centering
    \includegraphics[width=3.45in]{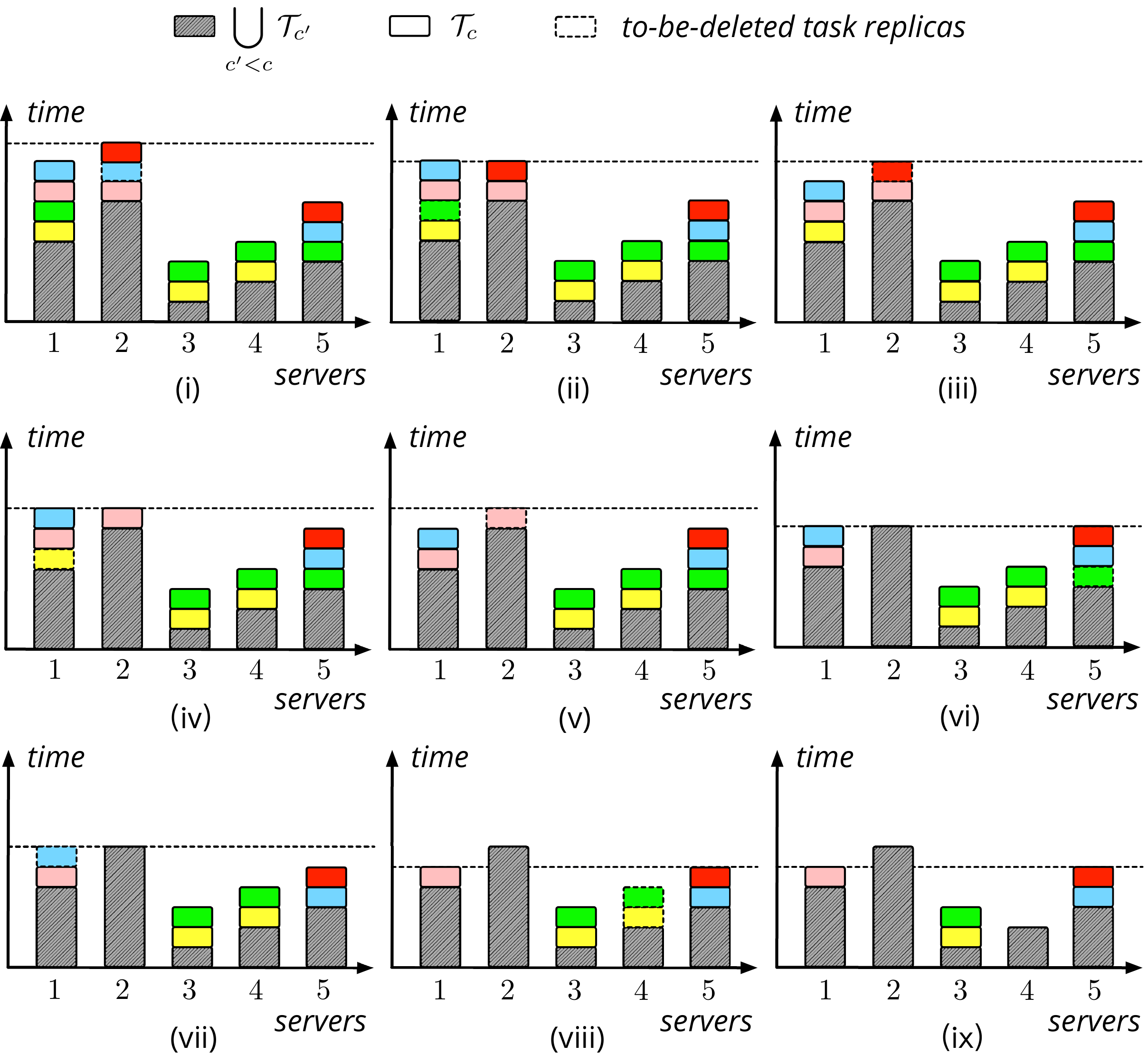}
    \caption{How RD works. (i)–(ix) demonstrate how to delete task replicas. In this example, $ \mu_m^c \equiv 1 $ for all the servers $ m \in \mathcal{M} $.}
    \label{rd}
\end{figure}

\textbf{Deletion phase.}
After initialization, RD enters the deletion phase to iteratively remove redundant task replicas while ensuring the estimated busy times of participating servers remain as balanced as possible. The steps in each deletion iteration are as follows.
\begin{enumerate}
    \item \textbf{Identify target server(s).} RD identifies the server(s) with the largest estimated busy time, referred to as the \textit{target server(s)} ($m^\star$). For example, in Fig.~\ref{rd}(i), server 2 is the target server in the first iteration.
    \item \textbf{Delete task replicas.} RD removes up to $\mu_{m^\star}^c$ task replicas from the target server to reduce its estimated busy time by one time slot. Among the tasks on the target server, the task with the largest number of replicas is prioritized for deletion. For example, in Fig.~\ref{rd}(i), server 2 hosts a red task with two replicas, a blue task with three replicas, and a pink task with two replicas. RD removes the blue task replica since it has the highest replication count.
    \item \textbf{Handle ties.} If multiple tasks have the same maximum number of replicas, RD selects the server with the largest initial estimated busy time (i.e., before job $c$’s tasks are assigned) as the target. For example, in Fig.~\ref{rd2}, servers 1 and 5 both have three replicas of the blue and green tasks, respectively. RD deletes the blue task replica from server 1 because server 1 has a higher initial estimated busy time.
    \item \textbf{Repeat until termination.} The deletion phase continues until all tasks in the target servers have only one replica. At this point, the completion time of job $c$ cannot be further reduced. For example, in Fig.~\ref{rd}(viii), RD terminates the deletion phase when all tasks in servers 1 and 5 have no other replicas.
\end{enumerate}
% After initialization, RD enters the deletion phase. To minimize the completion time of job $c$, RD removes redundant task replicas while keeping the estimated busy times of participating servers as balanced as possible. In each deletion iteration, RD first identifies the participating server(s) with the largest estimated busy time, referred to as the \textit{target server(s)}, denoted by $m^\star$. For example, in the first deletion iteration (Fig. \ref{rd}(i)), the target server is server 2. RD then removes \textit{up to} $\mu_{m^\star}^c$ task replicas from the target server to reduce its estimated busy time by one time slot. In each deletion iteration, RD removes task replicas with the largest numbers of copies.

\textbf{Finalization phase.} In the finalization phase, RD removes any remaining redundant task replicas to ensure that each task is assigned to only one server. This is done by continuing the replica removal process, focusing on servers with the highest estimated busy times. For example, in Fig.~\ref{rd}(ix), RD deletes the yellow and green task replicas from server 4 to balance workloads among the remaining participating servers.

% In the example in Fig. \ref{rd}, we assume $\mu_m^c \equiv 1$ for all the servers, so RD will remove at most one task replica from the target server. In Fig. \ref{rd}(i), at server 2, the red task has two replicas (in servers 2 and 5), the blue task has three replicas, and the pink task has two replicas. Therefore, RD deletes the blue task replica from server 2. In the second deletion iteration (Fig. \ref{rd}(ii)), the target servers are servers 1 and 2. In this case, the green task has the maximum number of replicas among all the tasks in these target servers. Thus, a green task replica is removed from server 1. In the third deletion iteration (Fig. \ref{rd}(iii)), the target server is server 2, and both the red task and the pink task have two replicas. In this case, RD randomly removes one of them. Figs. \ref{rd}(iv) to \ref{rd}(vii) show the subsequent deletion iterations.

% If there are multiple tasks with the same maximum number of replicas on distinct target servers, RD selects the server with the largest \textit{initial} estimated busy time, i.e., the estimated busy time of servers before assigning any tasks of job $c$, as the target server. For example, in Fig. \ref{rd2}, the target servers are servers 1 and 5, and both the blue task and the green task have three replicas. In this case, the blue task replica on server 1 will be deleted since server 1 has a larger initial estimated busy time.

\begin{figure}[h]
    \centering
    \includegraphics[width=2.16in]{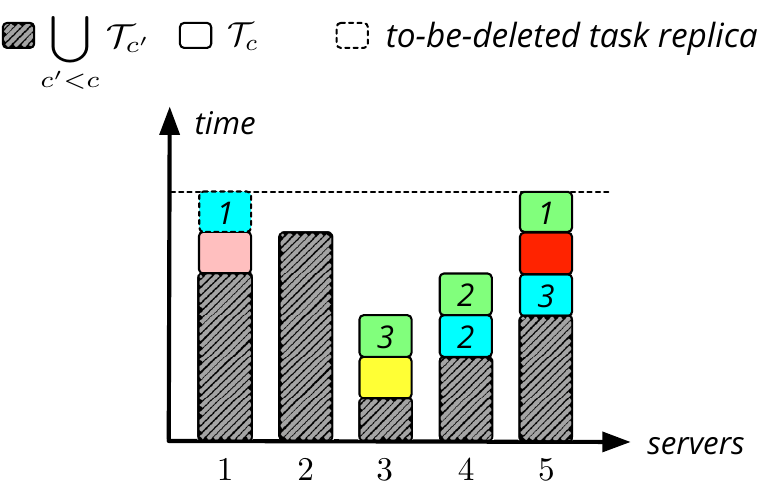}
    \caption{A target server with a larger initial estimated busy time has a task replica removed.}
    \label{rd2}
\end{figure}

% RD exits the deletion phase when all the tasks in the target server(s) have only one replica, indicating that the completion time of job $c$ cannot be reduced further. In Fig. \ref{rd}(viii), the deletion phase terminates since all the tasks in target servers 1 and 5 have no other replicas. In the final phase, RD removes redundant task replicas from the remaining participating servers so that each task is processed by only one server. This removal procedure operates similarly to the deletion phase, identifying and removing replicas from the server with the largest estimated busy time among the remaining servers. For example, in Fig. \ref{rd}(ix), RD deletes the replicas of the yellow and green tasks from server 4. This is done to keep the estimated busy times of all the servers as small and balanced as possible. 

RD differs from WF in its global perspective on task assignment. RD evaluates all available servers when deciding which task replica to delete, enabling it to achieve a more balanced workload distribution across the entire system. In contrast, WF operates locally by dividing tasks into groups and balancing workloads within each group during each iteration. While RD’s global approach often yields better performance, it also increases the computational complexity of the assignment process.
% RD assigns tasks to servers from a global perspective because it scans all available servers when selecting a task replica to delete. In contrast, WF minimizes the completion time of the arrived job locally by dividing tasks into task groups and balancing only the available servers within each task group in each iteration. Intuitively, RD can outperform WF in most cases, but RD's procedure is more complex to execute.

\subsubsection{Complexity Analysis}\label{s332}

In the implementation of RD, we use a priority queue to manage all servers, prioritizing them by their estimated busy times and breaking ties using their initial estimated busy times. Additionally, each server \( m \) maintains its own priority queue to manage the tasks assigned to it, prioritizing tasks by their number of replicas. In the worst-case scenario, each task is replicated across all servers, requiring \( \mathcal{O}(|\mathcal{T}_c| \cdot M) \) deletions. Each deletion operation affects the priority queues of all servers containing replicas of the deleted task. Updating the priority queue of up to \( M \) servers requires \( \mathcal{O}(M \cdot \log |\mathcal{T}_c|) \) time per deletion. Thus, the total time complexity of RD is:
$$
\mathcal{O} \Big( M^2 \cdot |\mathcal{T}_c| \cdot \log |\mathcal{T}_c| \Big).
$$
We experimentally compare the performance of RD with WF and other algorithms in Sec.~\ref{s5}. Theoretical performance analysis of RD is left as future work.

\section{Algorithms with Job Reordering}\label{s4}

This section considers an extended scenario in which outstanding job tasks queued on servers can be prioritized and reordered. By adjusting the execution order of outstanding jobs and reassigning their remaining tasks, the average completion time of jobs can be further reduced.

Inspired by \cite{hung2015scheduling} and \cite{9826037}, we adopt a similar approach to derive the job execution order, which emulates a \textit{shortest-remaining-time-first} (SRTF) policy for distributed job executions. When a new job $c$ arrives, reordering is triggered. Let $\mathcal{O}_c$ denote the set of outstanding jobs after the arrival of job $c$, and let $\mathcal{Q}_c$ represent the reordered list of outstanding jobs. The goal is to iteratively retrieve jobs from $\mathcal{O}_c$ and insert them into $\mathcal{Q}_c$ in a specific sequence until $|\mathcal{Q}_c| = |\mathcal{O}_c|$.

Assume there are $p$ jobs already sorted in $\mathcal{Q}_c$. For each job $c’ \in \mathcal{O}_c \backslash \mathcal{Q}_c$, we estimate its remaining time to completion, $\Phi_{c’}$, assuming it is the $(p+1)$-th job in the new order. This estimate, $\Phi_{c’}$, is derived using the WF algorithm, which considers the estimated busy times of servers, the current job order in $\mathcal{Q}_c$, and the unprocessed tasks of job $c’$.

Let $l$ represent the index of the job with the minimal estimated completion time $\Phi_l$ among all explored outstanding jobs. During the exploration of each job $c’$, we first compute a lower bound $\Phi_{c’}^-$ on the completion time, as defined in \eqref{xk} and \eqref{xk2} in Sec.~\ref{s312}. If $\Phi_{c’}^- \geq \Phi_l$, the exploration of job $c’$ terminates early. Otherwise, $\Phi_{c’}$ is computed using WF, and $l$ is updated if $\Phi_{c’} < \Phi_l$. This technique, known as \textit{early-exit}, significantly reduces the computation overhead for determining the $(p+1)$-th job in the new order. The rationale is that if $\Phi_{c’}^- \geq \Phi_l$, placing job $c’$ in the $(p+1)$ position will not result in a lower average job completion time.
% Assume that there are $ p $ sorted outstanding jobs in $ \mathcal{Q}_c $. For each job $ c' \in \mathcal{O}_c \backslash \mathcal{Q}_c $, we estimate the remaining time to complete job $ c' $ as $ \Phi_{c'} $, assuming that job $ c' $ is the $ (p+1) $-th job in the new order. Here, $ \Phi_{c'} $ is derived by WF with the estimated busy times of servers, given the assignment of $ p $ jobs in $ \mathcal{Q}_c $ and the number of unprocessed tasks of job $ c' $. We denote by $ l $ the index of the job with the minimal $ \Phi_l $ among all the outstanding jobs explored. When exploring each job $ c' \in \mathcal{O}_c \backslash \mathcal{Q}_c $, we first compute a lower bound $ \Phi_{c'}^- $ on the completion time of job $ c' $, as indicated by \eqref{xk} and \eqref{xk2} in Sec. \ref{s312}. If $ \Phi_{c'}^- \geq \Phi_l $, the exploration of job $ c' $ terminates. Otherwise, we calculate $ \Phi_{c'} $ by WF and update the best-explored job $ l $ if $ \Phi_{c'} < \Phi_l $. The reason is that if $ \Phi_{c'}^- > \Phi_l $, letting job $ c' $ be the $ (p+1) $-th job can never lead to a lower average job completion time. This technique, which significantly reduces the computation overhead on finding the $ (p+1) $-th job in the new order, is named \textit{early-exit}.

\begin{algorithm}[!h]
\caption{OCWF-ACC}
  {\color{black}\KwIn{Online arriving jobs $1, 2, ..., c, ...$, server capacities $\{ \mu_m^c \}_{m,c}$}}
  {\color{black}\KwOut{Assignment solution for each task of each arriving job}}
\label{ocwf-acc}
    \While{a new job $c$ arrives}{
        $\mathcal{O}_t \leftarrow \mathcal{O}_t \cup { c }$\\
        $\mathcal{Q}_t \leftarrow \emptyset$\\
        Initialize estimated busy times of servers: $\forall m \in \mathcal{M}$, $b_m^c \leftarrow 0$\\
        \While{$|\mathcal{Q}_c| < |\mathcal{O}_c|$}{
            $l \leftarrow \texttt{nullptr}$\\
            \For{each job $c’ \in \mathcal{O}_t \backslash \mathcal{Q}_t$}{
                Get unprocessed tasks of job $c’$ and calculate $\Phi_{c’}^-$ using \eqref{xk} and \eqref{xk2}\\
                \uIf{$l == \texttt{nullptr}$ \textbf{or} $\Phi_{c’}^- < \Phi_l$}{
                    Solve the task assignment problem $\mathcal{P}$ for job $c’$ using WF to calculate $\Phi_{c’}$\\
                    \If{$l == \texttt{nullptr}$ \textbf{or} $\Phi_{c’} < \Phi_l$}{
                        Derive the task assignment for $c’$ using WF\\
                        $l \leftarrow c’$
                    }
                }
                \Else{\textbf{break}{\hfill \tcp{early-exit}}}
            }
            $\mathcal{Q}_c \leftarrow \mathcal{Q}_c \cup \{ l \}$\\
            For each $m \in \mathcal{M}$, update the estimated busy time of server $m$
        }
    }
\end{algorithm}

% \begin{algorithm}[!h]
%     \caption{OCWF-ACC}
%     \label{ocwf-acc}
%     \While{a new job $c$ arrives}
%     {
%         $\mathcal{O}_t \leftarrow \mathcal{O}_t \cup \{ c \}$\\
%         $\mathcal{Q}_t \leftarrow \emptyset$\\
%         Initialize the estimated busy times of servers: $\forall m \in \mathcal{M}$, $b_m^c \leftarrow 0$\\
%         \While{$|\mathcal{Q}_c| < |\mathcal{O}_c|$}
%         {
%             $l \leftarrow \texttt{nullptr}$\\
%             \For{each job $c' \in \mathcal{O}_t \backslash \mathcal{Q}_t$}
%             {
%                 Get the unprocessed tasks of job $c'$, and calculate $\Phi_{c'}^-$ by \eqref{xk} and \eqref{xk2}\\
%                 \uIf{$l == \texttt{nullptr}$ \textbf{or} $\Phi_{c'}^- < \Phi_l$}
%                 {
%                     Get the unprocessed tasks of job $c'$, and solve the task assignment problem $\mathcal{P}$ by WF to get $\Phi_{c'}$\\
%                     \If{$l == \texttt{nullptr}$ \textbf{or} $\Phi_{c'} < \Phi_l$}
%                     {
%                         Derive the assignment of unprocessed tasks by WF\\
%                         $l \leftarrow c'$
%                     }
%                 }
%                 \Else{\textbf{break}{\hfill \tcp{early-exit}}}
%             }
%             $\mathcal{Q}_c \leftarrow \mathcal{Q}_c \cup \{ l \}$\\
%             For each $m \in \mathcal{M}$, update the estimated busy time of server $m$
%         }
%     }
% \end{algorithm}

We name the algorithm as OCWF-ACC (accelerated order-conscious WF). Algorithm \ref{ocwf-acc} presents the details of OCWF-ACC. Note that WF can be replaced by other task assignment algorithms.

\section{Experimental Results}\label{s5}

We conduct extensive simulations to verify the performance and efficiency of the proposed algorithms. We first present the simulation settings. Then, we show and discuss the results.

\subsection{Experimental Setup}\label{s6.1}

\begin{figure*}[ht]
  \centering
  \includegraphics[width=6.9in]{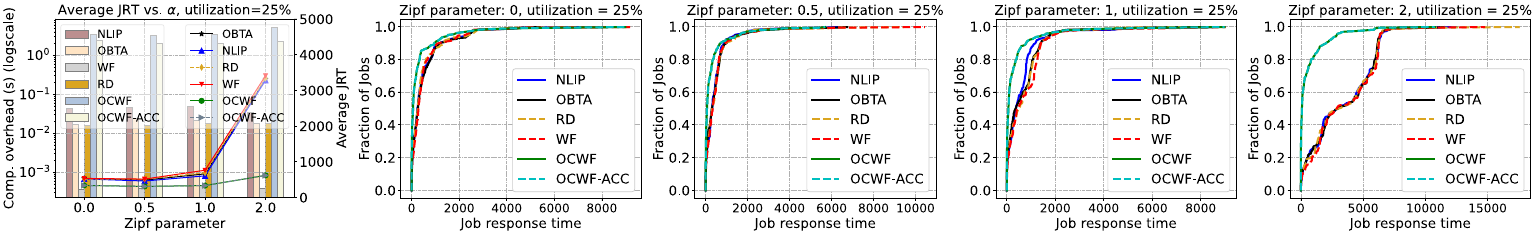}
  \caption{Average job completion time, computation overheads, and CDF of job completion times under 25\% system utilization for different algorithms.}
  \label{exp1}
\end{figure*}

\begin{figure*}[ht]
  \centering
  \includegraphics[width=6.9in]{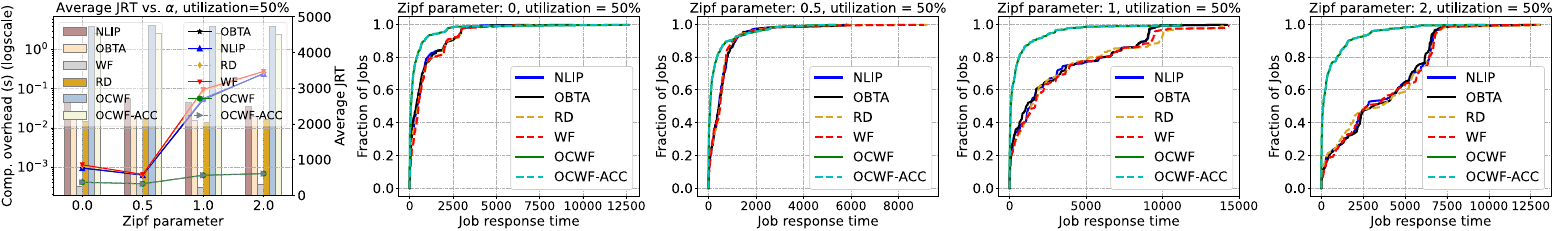}
  \caption{Average job completion time, computation overheads, and CDF of job completion times under 50\% system utilization for different algorithms.}
  \label{exp2}
\end{figure*}

\begin{figure*}[ht]
  \centering
  \includegraphics[width=6.9in]{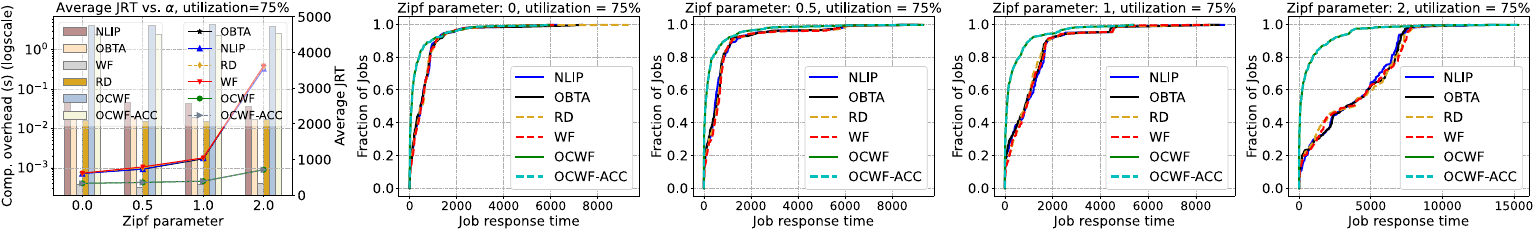}
  \caption{Average job completion time, computation overheads, and CDF of job completion times under 75\% system utilization for different algorithms.}
  \label{exp3}
\end{figure*}

\textbf{Job Traces.}
We used the Alibaba cluster trace dataset \cite{trace} for simulation. Specifically, a segment from the dataset \texttt{batch\_task.csv} in \texttt{cluster-trace-v2017} was selected, containing 250 jobs with a total of 113,653 tasks. Job arrivals and task durations were derived from the timestamps of recorded task events. The inter-arrival times of jobs were scaled to simulate system utilizations of 25\%, 50\%, and 75\%. The default number of servers was 100, with each server’s processing capacity for a given job randomly generated between 3 and 5 by default.

\textbf{Available Servers.} Each job entry in the trace corresponds to a task group. On average, each job contains 5.52 task groups. Data inputs for task groups were distributed across servers according to a Zipf distribution. For each task group, a random permutation of all servers was generated, and the task group was associated with the $i$-th server in the permutation with a probability proportional to $\frac{1}{i^\alpha}$. Here, $\alpha \in [0, 2]$ is the Zipf skew parameter, where $\alpha = 0$ indicates a uniform distribution. For a task group associated with server $m$, servers $m, m+1, \dots, m+p-1$ were chosen as its available servers, where $p$ was randomly generated between 8 and 12 by default.

\textbf{Algorithms.} 
We implemented six algorithms using DOcplex:\footnote{\url{https://ibmdecisionoptimization.github.io/docplex-doc/index.html}} NLIP, OBTA, WF, RD, OCWF, and OCWF-ACC. 
\begin{itemize}
    \item \textbf{NLIP} solves the non-linear program $\mathcal{P}$ directly for each job, without narrowing the search space or dividing it into subranges, as in OBTA.
    \item \textbf{OBTA} refines NLIP by reducing the search space for $\Phi_c$, improving efficiency while maintaining optimality  (see Sec. \ref{s31}).
    \item \textbf{WF} assigns tasks approximately in a water-filling manner to achieve a good balance between performance and overhead  (see Sec. \ref{s32}).
    \item \textbf{RD} employs a task assignment strategy based on an initialization-then-deletion process, aiming to iteratively balance server workloads (see Sec. \ref{s33}).
    \item \textbf{OCWF} reorders jobs to prioritize those with shorter remaining times but does not employ the early-exit optimization used in OCWF-ACC.
    \item \textbf{OCWF-ACC} extends OCWF with the early-exit technique to reduce computational overhead  (see Sec. \ref{s4}).
\end{itemize}
% We implement six algorithms with DOcplex\footnote{https://ibmdecisionoptimization.github.io/docplex-doc/index.html}: NLIP, OBTA, WF, RD, OCWF, and OCWF-ACC. NLIP differs from OBTA in that it solves the non-linear program $ \mathcal{P} $ for each job directly, without narrowing the search space of $\Phi_c$ and dividing it into subranges. For NLIP, OBTA, WF, and RD, we use the default FIFO policy for ordering job execution. OCWF is the non-accelerated version of OCWF-ACC, where the early-exit technique is not employed in the reordering of outstanding jobs (essentially the same as the implementation of the SWAG \cite{hung2015scheduling} and ATA-Greedy \cite{9826037} algorithms). The code is available for tests in other settings or large-scale validations.\footnote{https://github.com/hliangzhao/taos}

\textbf{Metrics.} We use the average job completion time of all jobs to measure performance and the computation overhead of each algorithm to measure efficiency.

The code is available online.\footnote{https://github.com/hliangzhao/taos}

\subsection{Performance Evaluation}\label{s6.2}
We compare the performance and efficiency of the algorithms by varying the Zipf parameter $\alpha$ from 0 to 2 and the system utilization from 25\% to 75\%. Figs. \ref{exp1}-\ref{exp3} present the results. In these figures, the right y-axis in the first subfigure indicates the average job completion time, while the left y-axis shows the log-scaled average computation overhead per job arrival. The other four subfigures visualize the cumulative distribution of job completion times for different $\alpha$ values. The key findings are summarized as follows.

\textit{Compared with NLIP, OBTA achieves a significant improvement in efficiency.} From Figs. \ref{exp1}-\ref{exp3}, we observe that OBTA and NLIP have fairly close performance of job completion time since both are theoretically optimal for balanced task assignment. By narrowing the search space and dividing it into subranges, OBTA reduces the computation overhead by nearly half compared to NLIP. This verifies our contribution in Sec. \ref{s32}. The notable reduction in computational overhead without sacrificing performance demonstrates the scalability of OBTA. This is especially relevant for large-scale systems where tasks arrive dynamically, and the computational cost of solving NLIP becomes a bottleneck. OBTA, by efficiently pruning the search space, provides a more practical approach for real-world applications, offering significant gains in execution speed while maintaining optimal task assignment.

\textit{WF closely approximates balanced task assignment with extremely low overhead.} WF performs close to OBTA and NLIP at almost all percentiles of job completion times, and its computation overhead is two orders of magnitude lower than OBTA’s. The ultra-low overhead offers a significant advantage and facilitates job reordering. The lightweight nature of WF makes it particularly appealing in scenarios where the frequency of task arrivals is high, and real-time decisions are necessary. Although WF sacrifices some degree of optimality, its approximation is sufficiently close to the theoretical optima for most practical purposes. The low computational cost not only accelerates the scheduling process but also allows it to be integrated effectively into more complex frameworks like OCWF and OCWF-ACC.

\textit{RD generally produces better task assignments than WF.} Our results show that RD generally outperforms WF. This observation is detailed in Table \ref{tab}, which presents the performance of each algorithm when $\alpha = 2$ and system utilization is at 75\%, as shown in the first figure of Fig.~\ref{exp4}. However, it is noteworthy that RD has a higher computation overhead than WF, which aligns with our complexity analysis. The time complexity of RD is quantified as $\mathcal{O}(M^2 \cdot n \log n)$, which is higher than WF’s $\mathcal{O}(KM \cdot \log n)$, where $M$ is the number of servers, $n$ is the number of tasks, and $K$ is the number of task groups in a job. RD’s global perspective on load balancing allows it to achieve better overall performance than WF in terms of job completion times. However, its relatively higher computational overhead makes it less suitable for latency-sensitive applications. Despite this, RD’s balancing strategy is particularly effective in systems with heterogeneous resources or non-uniform task distributions, where WF might underperform due to its localized approach.
% Our results show that RD generally outperforms WF. This observation is detailed in Table \ref{tab}, which presents the performance of each algorithm when $\alpha = 2$ and system utilization is at 75\%, as shown in the first figure of Fig. \ref{exp4}. However, it is noteworthy that RD has a higher computation overhead than WF, which aligns with our complexity analysis. The time complexity of RD is quantified as $\mathcal{O} (M^2 \cdot n \log n)$, which is higher than WF’s $\mathcal{O} (KM \cdot \log n)$, where $M$ is the number of servers, $n$ is the number of tasks, and $K$ is the number of task groups in a job.

\begin{table}[htbp]
\centering
\caption{\label{key}Average job completion time vs. \#available servers.}   
\begin{tabular}{ccccccc}
\toprule
\multirow{2}{*}{\textbf{Algorithms}} & \multicolumn{6}{c}{\textbf{Number of available servers}}                                                                                                   \\
                            & 4                     & 6                     & 8                     & 10                    & 12                 & Aver. \\ \midrule
OBTA                        & \multicolumn{1}{r}{12551} & \multicolumn{1}{r}{5814} & \multicolumn{1}{r}{4217} & \multicolumn{1}{r}{3609} & \multicolumn{1}{r}{3160} & \multicolumn{1}{r}{5870} \\
NLIP                        & \multicolumn{1}{r}{12468} & \multicolumn{1}{r}{5924} & \multicolumn{1}{r}{4394} & \multicolumn{1}{r}{3547} & \multicolumn{1}{r}{3164} & \multicolumn{1}{r}{5899}  \\
WF                          & \multicolumn{1}{r}{12872} & \multicolumn{1}{r}{6125} & \multicolumn{1}{r}{4408} & \multicolumn{1}{r}{3617} & \multicolumn{1}{r}{3188} & \multicolumn{1}{r}{6042} \\
RD                          & \multicolumn{1}{r}{12857} & \multicolumn{1}{r}{5987} & \multicolumn{1}{r}{4261} & \multicolumn{1}{r}{3595} & \multicolumn{1}{r}{3150} &  \multicolumn{1}{r}{5970}  \\
OCWF                        & \multicolumn{1}{r}{1699} & \multicolumn{1}{r}{958} & \multicolumn{1}{r}{765} & \multicolumn{1}{r}{712} & \multicolumn{1}{r}{657} &  \multicolumn{1}{r}{958}  \\
OCWF-ACC                    & \multicolumn{1}{r}{1699} & \multicolumn{1}{r}{958} & \multicolumn{1}{r}{765} & \multicolumn{1}{r}{712} & \multicolumn{1}{r}{657} &  \multicolumn{1}{r}{958}  \\ \bottomrule
\end{tabular}
\label{tab}
\end{table}

\textit{OCWF-ACC and OCWF are robust to skewness in data availability.} As shown in Figs.~\ref{exp1}-\ref{exp3}, when the skew parameter $\alpha$ changes from 0 to 2, the average job completion time achieved by FIFO-based algorithms (NLIP, OBTA, WF, and RD) shows a clear upward trend. In contrast, OCWF and OCWF-ACC do not exhibit a significant increase in average job completion time. This demonstrates that reordering and reallocation are crucial for handling skewness in data availability. The robustness of OCWF and OCWF-ACC against skewed data distributions underscores the importance of reordering in mitigating the impact of resource contention. By dynamically adjusting job priorities based on estimated completion times, these algorithms effectively balance workloads across servers, even under highly non-uniform conditions.
% As shown in Figs. \ref{exp1}-\ref{exp3}, when the skew parameter $\alpha$ changes from 0 to 2, the average job completion time achieved by FIFO-based algorithms (NLIP, OBTA, WF, and RD) shows a clear upward trend. In contrast, OCWF and OCWF-ACC do not exhibit a significant increase in average job completion time. This demonstrates that reordering and reallocation are crucial for handling skewness in data availability.

\textit{OCWF-ACC significantly accelerates OCWF.} The computation overhead of OCWF-ACC is only about half of OCWF’s. This verifies the advantage of the early-exit technique. The early-exit optimization not only reduces the time required for job reordering but also enhances the scalability of OCWF-ACC. This makes it a practical choice for systems with stringent latency requirements, as it balances the trade-off between computational efficiency and scheduling effectiveness.

\begin{figure*}[th]
  \centering
  \includegraphics[width=6.9in]{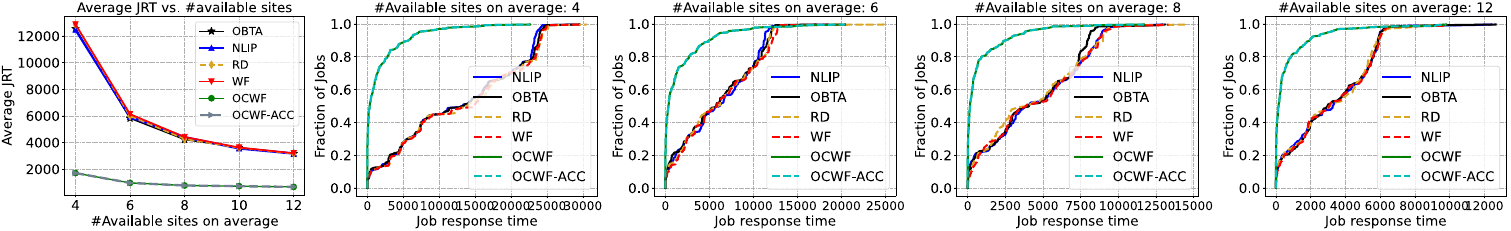}
  \caption{Average job completion time, and CDF of job completion times when $\alpha = 2$ under 75\% system utilization for different numbers of available servers.}
  \label{exp4}
\end{figure*}

\begin{figure*}[th]
  \centering
  \includegraphics[width=6.9in]{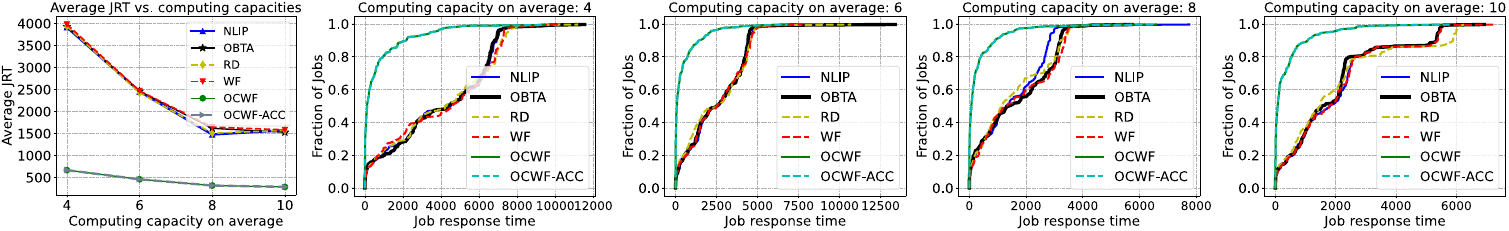}
  \caption{Average job completion time, and CDF of job completion times when $\alpha = 2$ under 75\% system utilization for different processing capacities.}
  \label{exp5}
\end{figure*}

Figs.~\ref{exp1}-\ref{exp3} illustrate that job completion times increase with system utilization. The performance trends discussed above are consistently observed across different levels of system utilization. Higher utilization levels amplify resource contention, which tests the scalability and robustness of the scheduling algorithms. While FIFO-based approaches struggle under these conditions, the reordering and reallocation capabilities of OCWF and OCWF-ACC enable them to maintain stable performance, even at high utilization.
% Figs. \ref{exp1}-\ref{exp3} show that job completion times increase with system utilization. The performance trends discussed above are consistently observed across different levels of system utilization.

Figs. \ref{exp4}-\ref{exp5} show the job completion times for varying numbers of available servers and profiled processing capacities, respectively. These experiments are conducted with $\alpha = 2$ and a system utilization of 75\%, as we are particularly interested in the performance when resource contention is high. From Fig. \ref{exp4}, we observe that a larger number of available servers decreases job completion times. More available servers provide greater flexibility in task assignment, allowing tasks to be distributed in a manner that better balances the load among servers and reduces job completion times. Fig.~\ref{exp5} demonstrates that increasing processing capacity also reduces job completion times. This is expected, as higher processing capacities enable more tasks to run in parallel, thereby decreasing job completion times. Notably, the relative performance of the algorithms remains largely unchanged across different numbers of available servers and processing capacities.
% Figs. \ref{exp4}-\ref{exp5} show the job completion times for varying numbers of available servers and profiled processing capacities, respectively. These experiments are conducted with $\alpha = 2$ and a system utilization of 75\%, as we are particularly interested in the performance when resource contention is high. From Fig. \ref{exp4}, we observe that a larger number of available servers decreases job completion times. More available servers provide greater flexibility in task assignment, allowing tasks to be distributed in a manner that better balances the load among servers and reduces job completion times. Fig. \ref{exp5} demonstrates that increasing processing capacity also reduces job completion times. This is expected, as higher processing capacities enable more tasks to run in parallel, thereby decreasing job completion times. Notably, the relative performance of the algorithms remains largely unchanged across different numbers of available servers and processing capacities.

{\color{black}
\subsection{Addition Evaluation on Algorithm Efficiency}\label{s6.3}
In this section, we evaluate the effectiveness of OBTA's two key optimization components: \textit{solution-space narrowing} and \textit{piecewise linearization}, by comparing OBTA with its variants. To isolate and evaluate the impact of each optimization strategy in OBTA, we implemented and compared it with three alternative approaches:
\begin{itemize}
    \item \textbf{LIP:} A direct solution to the linearized version of the original problem without any search-space reduction. The objective of the linearized problem $\mathcal{P}'$ is the same as $\mathcal{P}$, but the first set of non-linear constraints in \eqref{cons} are transformed into the following ones:
    \begin{align}
        \forall m \in \bigcup_{k \in \mathcal{K}_c} \mathcal{S}_c^k: 
        \left\{
        \begin{array}{l}
             \sum_{k \in \mathcal{K}_c} n_m^k \leq z_m, \\
             z_m \geq \Phi_c - b_m^c, \\
             z_m \geq 0, \\
             z_m \leq \Phi_c - b_m^c + \lambda \cdot y_{m,1}, \\
             z_m \leq \lambda \cdot y_{m,2}, \\
             y_{m,1} + y_{m,2} \leq 1, \\
             y_{m,1}, y_{m,2} \in \{ 0, 1 \},
        \end{array}
        \right.
    \end{align}
    where $\lambda$ is a sufficiently large number.
    \item \textbf{OBTA-N:} A variant of OBTA that removes the piecewise linearization component and retains only the solution-space narrowing mechanism.
    \item \textbf{OBTA-P:} A variant of OBTA that removes the solution-space narrowing step and applies full-range piecewise linearization.
\end{itemize}

\begin{table}[htbp]
    \color{black}
    \centering
    \caption{Computation overheads of algorithms to complete the scheduling of all jobs (in second).}
    \label{new_cmp_overhead}
    \begin{tabular}{ccccccc}
        \toprule
        \multirow{2}{*}{\textbf{Algorithms}} & \multicolumn{4}{c}{\textbf{Zipf Parameter $\alpha$}} \\
                                             & $\alpha=0$ & $\alpha=0.5$ & $\alpha=1$ & $\alpha=2$ \\ 
        \midrule
        OBTA                                 & \multicolumn{1}{r}{14.58} & \multicolumn{1}{r}{15.75} & \multicolumn{1}{r}{22.31} & \multicolumn{1}{r}{40.54} \\
        NLIP                                 & \multicolumn{1}{r}{21.60} & \multicolumn{1}{r}{24.42} & \multicolumn{1}{r}{33.89} & \multicolumn{1}{r}{48.17} \\
        LIP                                  & \multicolumn{1}{r}{17.79} & \multicolumn{1}{r}{17.89} & \multicolumn{1}{r}{27.82} & \multicolumn{1}{r}{42.06} \\
        OBTA-N                               & \multicolumn{1}{r}{18.59} & \multicolumn{1}{r}{18.97} & \multicolumn{1}{r}{25.49} & \multicolumn{1}{r}{43.96}     \\
        OBTA-P                               & \multicolumn{1}{r}{131.26} & \multicolumn{1}{r}{133.94} & \multicolumn{1}{r}{136.59} & \multicolumn{1}{r}{85.42} \\
        WF                                   & \multicolumn{1}{r}{10.12} & \multicolumn{1}{r}{9.63} & \multicolumn{1}{r}{16.16} & \multicolumn{1}{r}{37.25}     \\ 
        RD                                   & \multicolumn{1}{r}{12.48} & \multicolumn{1}{r}{10.98} & \multicolumn{1}{r}{15.93} & \multicolumn{1}{r}{37.79}     \\
        \bottomrule
    \end{tabular}
\end{table}

\begin{table}[htbp]
    \color{black}
    \centering
    \caption{Average job completion time achieved by algorithms.}
    \label{new_cmp_jrt}
    \begin{tabular}{ccccccc}
        \toprule
        \multirow{2}{*}{\textbf{Algorithms}} & \multicolumn{4}{c}{\textbf{Zipf Parameter $\alpha$}} \\
                                             & $\alpha=0$ & $\alpha=0.5$ & $\alpha=1$ & $\alpha=2$ \\ 
        \midrule
        OBTA                                 & \multicolumn{1}{r}{694} & \multicolumn{1}{r}{823} & \multicolumn{1}{r}{1517} & \multicolumn{1}{r}{7081} \\
        NLIP                                 & \multicolumn{1}{r}{709} & \multicolumn{1}{r}{813} & \multicolumn{1}{r}{1494} & \multicolumn{1}{r}{7164} \\
        LIP                                  & \multicolumn{1}{r}{682} & \multicolumn{1}{r}{798} & \multicolumn{1}{r}{1571} & \multicolumn{1}{r}{7183} \\
        OBTA-N                               & \multicolumn{1}{r}{704} & \multicolumn{1}{r}{831} & \multicolumn{1}{r}{1518} & \multicolumn{1}{r}{7151}     \\
        OBTA-P                               & \multicolumn{1}{r}{703} & \multicolumn{1}{r}{818} & \multicolumn{1}{r}{1548} & \multicolumn{1}{r}{7213} \\
        WF                                   & \multicolumn{1}{r}{755} & \multicolumn{1}{r}{861} & \multicolumn{1}{r}{1569} & \multicolumn{1}{r}{7270}     \\ 
        RD                                   & \multicolumn{1}{r}{772} & \multicolumn{1}{r}{852} & \multicolumn{1}{r}{1567} & \multicolumn{1}{r}{7252}     \\
        \bottomrule
    \end{tabular}
\end{table}

We evaluated all four algorithms: OBTA, LIP, OBTA-N, and OBTA-P, alongside WF, RD, and NLIP under the experimental settings demonstrated in Sec. \ref{s6.1} (the system utilization is set to 75\% by default). As shown in Tables \ref{new_cmp_overhead} and \ref{new_cmp_jrt}, in terms of solution quality, all algorithms perform similarly, indicating that the linearization and search-space reduction do not compromise optimality. However, in terms of computational efficiency: WF $>$ RD $>$ OBTA $>$ LIP $>$ OBTA-N $>$ NLIP $>$ OBTA-P, where `$>$' indicates lower computational cost (i.e., faster execution). These results demonstrate that both components of OBTA contribute to its efficiency. In particular, solution-space narrowing (OBTA-N) achieves the lowest runtime among all variants, suggesting that this component has the most significant impact on reducing computational overhead. Meanwhile, the combination of both strategies in OBTA achieves the best balance between efficiency and stability across different workloads.

\subsection{Robustness Under Large-Scale and Heterogeneous Settings}\label{s6.4}
To further assess scalability and robustness, we tested the algorithms under larger-scale scenarios and heterogeneous server environments.

First, we increased the input size to include 500 jobs, 256 servers, and 173,848 tasks, significantly expanding the solution space. The system utilization is set to 75\% by default. As shown in Table~\ref{new_large_scale}, the relative performance trends remain consistent, confirming that OBTA and its variants scale effectively with increasing problem size.

\begin{table*}[htbp]
    \color{black}
    \centering
    \caption{Large-scale experiments with 75\% system utilization and Zipf parameter $\alpha$ being 2.}
    \label{new_large_scale}
    \begin{tabular}{cccccccccc}
        \toprule
        \textbf{Algorithms}                      & OBTA & NLIP & LIP & OBTA-N & OBTA-P & WF & RD & OCWF & OCWF-ACC \\ 
        \midrule
        \textbf{Overhead (s)}  & 89 & 96 & 104 & 101 & 221 & 78 & 77 & 6538 & 3118 \\
        \textbf{Average job completion time}                     & 8639 & 8764 & 9059 & 8716 & 8878 & 9281 & 8873 & 1189 & 1234 \\
        \bottomrule
    \end{tabular}
\end{table*}

Next, we evaluated the algorithms under various server capacity distributions to simulate heterogeneous computing environments. Specifically, we considered uniform, power-law, log-normal, and weighted distributions. These models reflect real-world diversity in server capabilities. Tables~\ref{new_jrt_cap} and \ref{new_overhead_cap} summarize the average job completion time and computational overheads under these distributions. Across all settings, the relative rankings of algorithmic efficiency remain stable, demonstrating that the performance improvements achieved by OBTA are not sensitive to infrastructure heterogeneity. These experiments confirm that OBTA maintains its efficiency advantage even when scaling to large workloads, and the proposed algorithms perform consistently well under diverse server capacity distributions, showing their robustness to system heterogeneity.

\begin{table}[htbp]
    \color{black}
    \centering
    \caption{Average job completion time achieved by algorithms under different server capacity distributions.}
    \label{new_jrt_cap}
    \begin{tabular}{ccccccc}
        \toprule
        \multirow{2}{*}{\textbf{Algorithms}} & \multicolumn{4}{c}{\textbf{Types of Server Capacity Distribution}} \\
                                             & Uniform & Power-law & Log-normal & Weighted \\
        \midrule
        OBTA                                 & \multicolumn{1}{r}{6713} & \multicolumn{1}{r}{6530} & \multicolumn{1}{r}{7667} & \multicolumn{1}{r}{8872} \\
        NLIP                                 & \multicolumn{1}{r}{6935} & \multicolumn{1}{r}{6803} & \multicolumn{1}{r}{7833} & \multicolumn{1}{r}{8987} \\
        LIP                                  & \multicolumn{1}{r}{6865} & \multicolumn{1}{r}{6433} & \multicolumn{1}{r}{7878} & \multicolumn{1}{r}{9062} \\
        OBTA-N                               & \multicolumn{1}{r}{7076} & \multicolumn{1}{r}{6778} & \multicolumn{1}{r}{7602} & \multicolumn{1}{r}{8899} \\
        OBTA-P                               & \multicolumn{1}{r}{6821} & \multicolumn{1}{r}{6549} & \multicolumn{1}{r}{7622} & \multicolumn{1}{r}{8812} \\
        WF                                   & \multicolumn{1}{r}{6966} & \multicolumn{1}{r}{6599} & \multicolumn{1}{r}{7673} & \multicolumn{1}{r}{8846} \\
        RD                                   & \multicolumn{1}{r}{7071} & \multicolumn{1}{r}{6819} & \multicolumn{1}{r}{7691} & \multicolumn{1}{r}{8856} \\
        OCWF                                 & \multicolumn{1}{r}{1172} & \multicolumn{1}{r}{1118} & \multicolumn{1}{r}{1634} & \multicolumn{1}{r}{1674} \\
        OCWF-ACC                             & \multicolumn{1}{r}{1222} & \multicolumn{1}{r}{1176} & \multicolumn{1}{r}{1634} & \multicolumn{1}{r}{1674} \\
        \bottomrule
    \end{tabular}
\end{table}

\begin{table}[htbp]
    \color{black}
    \centering
    \caption{Computation overheads of algorithms under different server capacity distributions (in second).}
    \label{new_overhead_cap}
    \begin{tabular}{ccccccc}
        \toprule
        \multirow{2}{*}{\textbf{Algorithms}} & \multicolumn{4}{c}{\textbf{Types of Server Capacity Distribution}} \\
                                             & Uniform & Power-law & Log-normal & Weighted \\
        \midrule
        OBTA                                 & \multicolumn{1}{r}{36.61} & \multicolumn{1}{r}{34.27} & \multicolumn{1}{r}{50.11} & \multicolumn{1}{r}{54.37} \\
        NLIP                                 & \multicolumn{1}{r}{39.97} & \multicolumn{1}{r}{39.35} & \multicolumn{1}{r}{57.73} & \multicolumn{1}{r}{57.08} \\
        LIP                                  & \multicolumn{1}{r}{40.75} & \multicolumn{1}{r}{36.88} & \multicolumn{1}{r}{52.26} & \multicolumn{1}{r}{56.03} \\
        OBTA-N                               & \multicolumn{1}{r}{42.09} & \multicolumn{1}{r}{38.06} & \multicolumn{1}{r}{51.96} & \multicolumn{1}{r}{56.84} \\
        OBTA-P                               & \multicolumn{1}{r}{85.23} & \multicolumn{1}{r}{75.35} & \multicolumn{1}{r}{99.38} & \multicolumn{1}{r}{103.93} \\
        WF                                   & \multicolumn{1}{r}{32.63} & \multicolumn{1}{r}{29.22} & \multicolumn{1}{r}{43.16} & \multicolumn{1}{r}{45.95} \\
        RD                                   & \multicolumn{1}{r}{34.76} & \multicolumn{1}{r}{31.67} & \multicolumn{1}{r}{46.03} & \multicolumn{1}{r}{47.05} \\
        OCWF                                 & \multicolumn{1}{r}{578.82} & \multicolumn{1}{r}{725.53} & \multicolumn{1}{r}{769.46} & \multicolumn{1}{r}{780.93} \\
        OCWF-ACC                             & \multicolumn{1}{r}{410.00} & \multicolumn{1}{r}{399.36} & \multicolumn{1}{r}{512.85} & \multicolumn{1}{r}{508.66} \\
        \bottomrule
    \end{tabular}
\end{table}

To test the sensitivity of our approaches to how tasks map to available servers, we modified the distribution of server availability. Specifically, we replaced the original Zipf-based assignment with log-normal, exponential, and uniform distributions. As shown in Tables \ref{as_jrt} and \ref{as_overhead}, the experimental results further support our initial conclusions, indicating that the performance trends of our algorithms are largely insensitive to the specific distribution patterns.

\begin{table}[htbp]
    \color{black}
    \centering
    \caption{Average job completion time achieved by algorithms under different available server distributions.}
    \label{as_jrt}
    \begin{tabular}{ccccccc}
        \toprule
        \multirow{2}{*}{\textbf{Algorithms}} & \multicolumn{4}{c}{\textbf{Types of Available Server Distribution}} \\
                                             & Zipf & Log-normal & Exponential & Uniform \\
        \midrule
        OBTA                                 & \multicolumn{1}{r}{6353} & \multicolumn{1}{r}{800} & \multicolumn{1}{r}{783} & \multicolumn{1}{r}{767} \\
        NLIP                                 & \multicolumn{1}{r}{6513} & \multicolumn{1}{r}{803} & \multicolumn{1}{r}{780} & \multicolumn{1}{r}{745} \\
        LIP                                  & \multicolumn{1}{r}{6300} & \multicolumn{1}{r}{790} & \multicolumn{1}{r}{804} & \multicolumn{1}{r}{784} \\
        OBTA-N                               & \multicolumn{1}{r}{6410} & \multicolumn{1}{r}{770} & \multicolumn{1}{r}{777} & \multicolumn{1}{r}{747} \\
        OBTA-P                               & \multicolumn{1}{r}{6486} & \multicolumn{1}{r}{816} & \multicolumn{1}{r}{779} & \multicolumn{1}{r}{740} \\
        WF                                   & \multicolumn{1}{r}{6529} & \multicolumn{1}{r}{825} & \multicolumn{1}{r}{804} & \multicolumn{1}{r}{771} \\
        RD                                   & \multicolumn{1}{r}{6750} & \multicolumn{1}{r}{818} & \multicolumn{1}{r}{806} & \multicolumn{1}{r}{817} \\
        OCWF                                 & \multicolumn{1}{r}{1197} & \multicolumn{1}{r}{345} & \multicolumn{1}{r}{345} & \multicolumn{1}{r}{311} \\
        OCWF-ACC                             & \multicolumn{1}{r}{1204} & \multicolumn{1}{r}{339} & \multicolumn{1}{r}{342} & \multicolumn{1}{r}{333} \\
        \bottomrule
    \end{tabular}
\end{table}

\begin{table}[htbp]
    \color{black}
    \centering
    \caption{Computation overheads of algorithms under different available server distributions (in second).}
    \label{as_overhead}
    \begin{tabular}{ccccccc}
        \toprule
        \multirow{2}{*}{\textbf{Algorithms}} & \multicolumn{4}{c}{\textbf{Types of Available Server Distribution}} \\
                                             & Zipf & Log-normal & Exponential & Uniform \\
        \midrule
        OBTA                                 & \multicolumn{1}{r}{32.30} & \multicolumn{1}{r}{17.07} & \multicolumn{1}{r}{14.82} & \multicolumn{1}{r}{15.80} \\
        NLIP                                 & \multicolumn{1}{r}{35.93} & \multicolumn{1}{r}{22.85} & \multicolumn{1}{r}{21.25} & \multicolumn{1}{r}{19.60} \\
        LIP                                  & \multicolumn{1}{r}{36.14} & \multicolumn{1}{r}{19.94} & \multicolumn{1}{r}{17.29} & \multicolumn{1}{r}{18.37} \\
        OBTA-N                               & \multicolumn{1}{r}{36.17} & \multicolumn{1}{r}{20.17} & \multicolumn{1}{r}{18.61} & \multicolumn{1}{r}{19.48} \\
        OBTA-P                               & \multicolumn{1}{r}{76.16} & \multicolumn{1}{r}{117.04} & \multicolumn{1}{r}{110.96} & \multicolumn{1}{r}{107.29} \\
        WF                                   & \multicolumn{1}{r}{27.64} & \multicolumn{1}{r}{8.75} & \multicolumn{1}{r}{8.39} & \multicolumn{1}{r}{10.35} \\
        RD                                   & \multicolumn{1}{r}{31.21} & \multicolumn{1}{r}{10.21} & \multicolumn{1}{r}{9.89} & \multicolumn{1}{r}{9.98} \\
        OCWF                                 & \multicolumn{1}{r}{731.56} & \multicolumn{1}{r}{569.83} & \multicolumn{1}{r}{574.69} & \multicolumn{1}{r}{559.44} \\
        OCWF-ACC                             & \multicolumn{1}{r}{408.20} & \multicolumn{1}{r}{419.25} & \multicolumn{1}{r}{432.90} & \multicolumn{1}{r}{426.22} \\
        \bottomrule
    \end{tabular}
\end{table}
}

\section{Related Work}\label{s6}

Job scheduling has been extensively studied from both theoretical \cite{gautam2015survey,BSP,8486422,attiya2020job,zhang2020evolving,liang2020data,narayanan2020heterogeneity,dpe,10433234} and system-level perspectives \cite{schedule1,199394,222631,199388,borg2,decima}. Theoretical works often formulate job completion time minimization as combinatorial, constrained optimization problems, addressing them using a variety of techniques, particularly approximation algorithms. These studies provide foundational insights into optimality and computational tractability, often emphasizing complexity bounds and performance guarantees.

A critical system-level consideration in distributed scheduling is data locality, which emphasizes assigning tasks close to where the required data chunks reside to minimize communication overhead. Rooted in the MapReduce paradigm \cite{guo2012investigation}, this principle has driven numerous scheduling algorithms for both homogeneous and heterogeneous systems \cite{zhang2011improving,wang2014maptask,naik2019data,beaumont2020performance}. These approaches typically aim to balance data locality against other performance goals such as load balancing \cite{wang2014optimizing}, throughput \cite{xie2016scheduling}, delay \cite{naik2019data,fu2020optimal}, fairness \cite{10.1145/3337821.3337843}, and job completion time \cite{dong2022smart,hung2015scheduling,9488810,added1,added2,added3}.

Within this context, scheduling strategies that explicitly preserve data locality in multi-task, online job scenarios have received increased attention \cite{hung2015scheduling,8056949,10.1145/3337821.3337843,beaumont2020performance,9826037}. For instance, Hung \textit{et al.} \cite{hung2015scheduling} proposed SWAG, a workload-aware greedy algorithm that schedules jobs by estimating the completion times of all outstanding jobs. SWAG was the first to formalize job reordering in distributed environments with data locality constraints.

\begin{table*}[ht]
    \color{black}
    \centering
    \caption{Comparison of related works for distributed job executions.}
    \label{tab:related_comparison}
    \begin{tabular}{lcccccc}
    \toprule
    \textbf{Related Work} & \textbf{Online} & \textbf{Data Locality} & \textbf{Multi-task Job} & \textbf{Theoretical Analysis} & \textbf{Real Trace Evaluation} \\
    \midrule
    SWAG \cite{hung2015scheduling} & \ding{51} & \ding{55} & \ding{51} & \ding{55} & \ding{51} \\
    BTAaJ / ATA-Greedy \cite{9826037} & \ding{51} & \ding{51} & \ding{51} & \ding{55} & \ding{51} \\
    SG-PBFS / EG-SJF \cite{added1,added3} & \ding{55} & \ding{55} & \ding{55} & \ding{55} & \ding{55} \\
    MTD-DHJS \cite{added2} & \ding{55} & \ding{55} & \ding{55} & \ding{55} & \ding{55} \\
    \textbf{This Work} & \ding{51} & \ding{51} & \ding{51} & \ding{51} & \ding{51} \\
    \bottomrule
    \end{tabular}
\end{table*}

Building on this, Guan \textit{et al.} \cite{9826037} proposed the BTAaJ algorithm, which significantly improved upon SWAG by utilizing a well-constructed flow network to assign tasks for newly arrived jobs. BTAaJ groups the tasks of a job based on their available servers, achieving load balancing by minimizing a control parameter $C$. Additionally, Guan \textit{et al.} proposed the ATA-Greedy algorithm, which enhances scheduling performance by reordering outstanding jobs. This reordering process builds on the approach in \cite{hung2015scheduling}, incorporating the control parameter $C$ to refine completion time estimations. Our work extends these efforts by introducing algorithms that offer both improved performance and reduced computational overhead, making significant advancements over \cite{hung2015scheduling} and \cite{9826037}.

Other dimensions of distributed job scheduling, such as fairness and data movement, have also been explored. For instance, Guan \textit{et al.} \cite{10.1145/3337821.3337843} extended the concept of max-min fairness from single-machine settings to distributed job executions over multiple servers. Their Aggregate Max-min Fairness policy is designed to be Pareto-efficient, envy-free, and strategy-proof, providing equitable resource allocation in distributed systems. Beaumont \textit{et al.} \cite{beaumont2020performance} addressed the \textsc{Comm-Oriented} problem, which optimizes the makespan while minimizing data movement. However, their focus on single-job scenarios limits its applicability to multi-job online scheduling, which is the primary concern of our work. \textcolor{black}{Several recent works have also proposed heuristic scheduling algorithms for cloud environments. For instance, Banerjee \textit{et al.} \cite{added2} proposed MTD-DHJS, which combines Johnson sequencing and Round Robin to reduce makespan on a fixed three-server architecture. However, their model assumes batch-style job processing and lacks support for task-level data locality, complex job structures, or online arrivals. Similarly, Murad \textit{et al.} \cite{added1,added3} developed SG-PBFS and EG-SJF, which enhance traditional priority-rule (PR) scheduling with gap-based backfilling to improve flow time and delay. These methods rely on metadata such as CPU time and arrival time for decision-making, but diverge from our work in both modeling assumptions and distributed settings. In particular, they do not handle replicated input data, multi-task jobs, or theoretical performance guarantees.}

\section{Conclusion}\label{s7}

In this paper, we explored the task assignment problem in the context of online distributed job executions, with data locality preserved. Our contributions began with formulating the task assignment problem as a non-linear integer program. We developed the OBTA algorithm, known for its efficiency in reducing computational overhead by narrowing the search space of potential solutions. Additionally, we extended the WF heuristic, providing an in-depth, nontrivial tight analysis of its approximation factor, and introduced the innovative RD algorithm, which demonstrates enhanced performance. Our investigation into job reordering, particularly with the implementation of the novel early-exit technique, effectively optimized average job completion times while maintaining low computational overheads. The effectiveness of all proposed algorithms was validated through extensive trace-driven simulations. Both OBTA and the early-exit technique showed substantial reductions in computational overhead. Meanwhile, WF emerged as a practical, lower-overhead alternative for large-scale scenarios. Despite its greater computational intensity compared to WF, RD exhibited superior performance of job completion time. Looking ahead, integrating these algorithms into diverse computational settings and adapting them to the evolving dynamics of distributed job executions hold potential for further enhancing their applicability.

% \section*{Acknowledgments}

% This work was supported in part by the Singapore Ministry of Education under Academic Research Fund Tier 2 Award MOE-T2EP20121-0005, the National Natural Science Foundation of China under Grants 62125206 and U20A20173, the National Key Research and Development Program of China under Grant 2022YFB4500100, and the Key Research Project of Zhejiang Province under Grant 2022C01145.

\appendices
\section{Proof of Theorem \ref{theorem-wf-ub}}\label{app_proof}
We need to prove that $\textsc{WF}(I) \leq K_c \cdot \textsc{OPT}(I)$ for any possible arrival instance $I$. For simplicity, we omit $I$ in the notation and use ALG to represent both the algorithm and the estimated completion time of job $c$ under that algorithm. ALG can refer to $\textsc{OPT}$, $\textsc{WF}$, $\textsc{WF}_k$, or $\textsc{OPT}_k$ (the last two notations will be introduced shortly).
% We need to prove that $ \textsc{WF}(I) \leq K_c \cdot \textsc{OPT}(I) $ holds for any possible arrival instance $ I $. Hereafter, we drop $ I $ for simplicity and use ALG to represent the algorithm itself or the estimated completion time of job $ c $ by applying the algorithm interchangeably (ALG can be OPT, WF, WF$_k$, and OPT$_k$, where the last two notations will be introduced soon).

Remember that when WF assigns tasks in group $ k $, we call server $ m \in \mathcal{S}_c^k $ a participating server if $ \xi_k - b_m^c(k-1) > 0 $. We use $ \Omega_c^k \subseteq \mathcal{S}_c^k $ to denote this set of participating servers. We denote by $ \textsc{WF}_k $ the number of time slots required to \textit{complete} the processing of all the tasks in $ \bigcup_{k' \leq k} \mathcal{T}_c^{k'} $.\footnote{$\textsc{WF}_k$ is the number of time slots between job $c$'s arrival and the time when the last task in $ \mathcal{T}_c^k $ completes.} Note that for each server $ m \in \bigcup_{k' \leq k} \mathcal{S}_c^{k'} $ with FIFO queues, the tasks issued by job $ c $ can be processed only after the outstanding tasks of previous jobs are processed. Thus, by definition, for each $ k \in \mathcal{K}_c $,
\begin{equation}
    \textsc{WF}_k = \max_{m \in \bigcup_{k' \leq k} \Omega_c^{k'}} b_m^c(k),
    \label{eq_042503}
\end{equation}
and for each $ k < K_c $,
\begin{equation}
    \textsc{WF}_{k+1} = \max \left\{ \textsc{WF}_k, \max_{m \in \Omega_c^{k+1}} b_m^c(k+1) \right\}.
    \label{wf_delta}
\end{equation}
Recall that $ b_m^c(k) $ is updated by \eqref{bk_update}. Fig. \ref{wf_k} visualizes $ \textsc{WF}_k $.

\begin{figure}[h]
  \centering
  \includegraphics[width=2.7in]{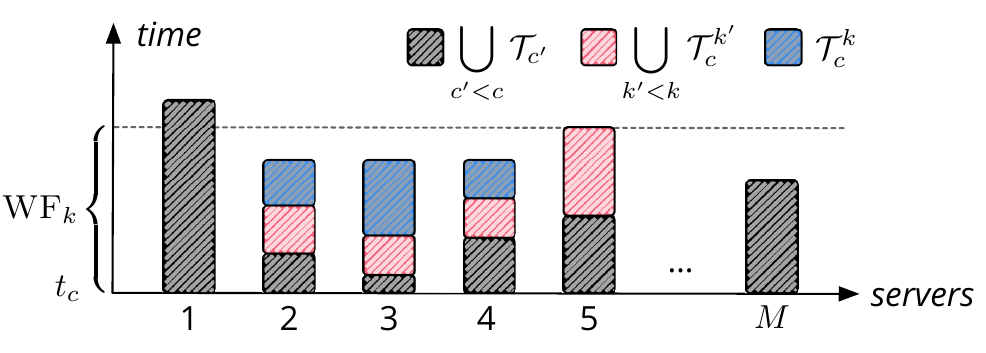}
  \caption{Visualization of $\textsc{WF}_k$. Note that $\textsc{WF}_{k} = \textsc{WF}_{k-1}$ for the constructed example.}
  \label{wf_k}
\end{figure}

Since $ \mathcal{T}_c^1 \subseteq \mathcal{T}_c $, $ \textsc{WF}_1 \leq \textsc{OPT} $. Note that $ \textsc{WF} = \textsc{WF}_{K_c} $. If we can prove that for each $ k < K_c $,
\begin{equation}
    \textsc{WF}_{k+1} - \textsc{WF}_k \leq \textsc{OPT},
    \label{delta}
\end{equation}
then we have
\begin{equation}
    \textsc{WF} = \textsc{WF}_1 + \sum_{k < K_c} (\textsc{WF}_{k+1} - \textsc{WF}_k) \leq K_c \cdot \textsc{OPT},
\end{equation}
which is exactly our target. In the following, we will show that \eqref{delta} indeed holds. To do so, we exploit the optimal assignment of the tasks of a job \textit{containing the tasks $ \mathcal{T}_c^{k+1} $ only}. Similar to $ \textsc{WF}_{k+1} $, we denote by $ \textsc{OPT}_{k+1} $ the number of time slots required to complete the processing of tasks in $ \mathcal{T}_c^{k+1} $. It is easy to see that
\begin{equation}
    \textsc{OPT}_{k+1} \leq \textsc{OPT}
    \label{opt_k}
\end{equation}
holds for each $ k \in \mathcal{K}_c $, since $ \mathcal{T}_c^{k+1} \subseteq \mathcal{T}_c $.

To prove \eqref{delta}, we consider different cases of the assignment of tasks in $\mathcal{T}_c^{k+1}$ by WF separately.

\textbf{Case I.} $\mathcal{S}_c^{k+1} \cap \bigcup_{k' \leq k} \Omega_c^{k'} = \emptyset$. In this case, the available servers for group $ k+1 $ do not overlap with the participating servers for groups $ 1, \ldots, k $. Thus, by the definition of WF, such an assignment must be optimal if a job contains only the tasks $ \mathcal{T}_c^{k+1} $. Hence, by \eqref{opt_k}, we have
\begin{equation}
    \max_{m \in \Omega_c^{k+1}} b_m^c(k+1) = \textsc{OPT}_{k+1} \leq \textsc{OPT}.
    \label{eq_042501}
\end{equation}

\begin{itemize}
    \item As visualized in Fig. \ref{new-ana-1}(i), if $ \max_{m \in \Omega_c^{k+1}} b_m^c(k+1) \leq \textsc{WF}_k $, by \eqref{wf_delta}, we have $ \textsc{WF}_{k+1} = \textsc{WF}_k $. Thus,
    \begin{equation}
        \textsc{WF}_{k+1} - \textsc{WF}_k = 0 < \textsc{OPT}.
    \end{equation}
    
    \item As visualized in Fig. \ref{new-ana-1}(ii), if $ \max_{m \in \Omega_c^{k+1}} b_m^c(k+1) > \textsc{WF}_k $, by \eqref{wf_delta} and \eqref{eq_042501}, we have
    \begin{align}
        \textsc{WF}_{k+1} - \textsc{WF}_k &< \textsc{WF}_{k+1} \nonumber\\
        &= \max_{m \in \Omega_c^{k+1}} b_m^c(k+1) \nonumber\\
        &\leq \textsc{OPT}.
    \end{align}
\end{itemize}

\begin{figure}[h]
  \centering
  \includegraphics[width=2.88in]{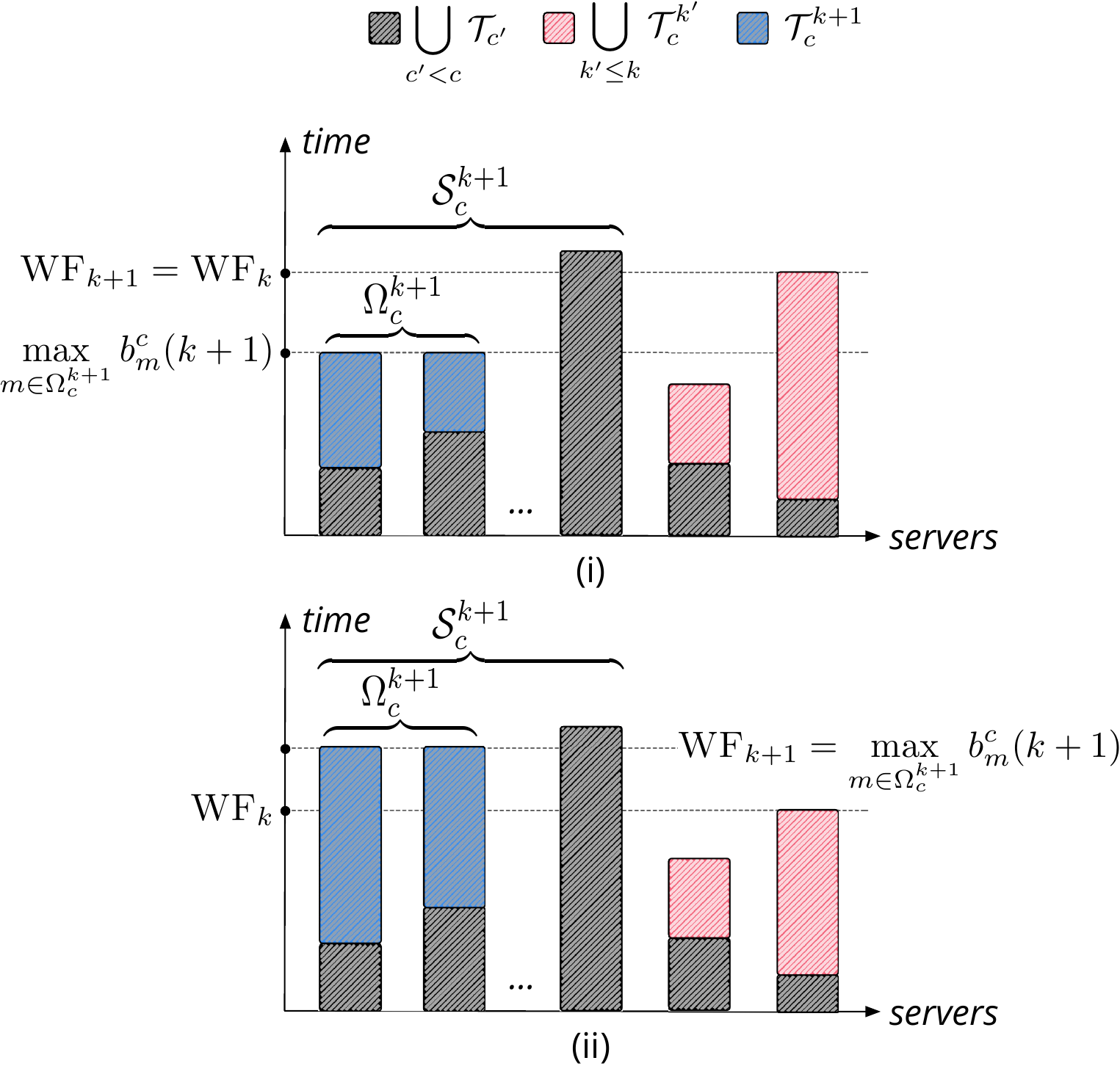}
  \caption{Visualization of Case I. In (i), $\max_{m \in \Omega_c^{k+1}} b_m^c(k+1) \leq \textsc{WF}_k$. In (ii), $\max_{m \in \Omega_c^{k+1}} b_m^c(k+1) > \textsc{WF}_k$.}
  \label{new-ana-1}
\end{figure}

\textbf{Case II.} $\mathcal{S}_c^{k+1} \cap \bigcup_{k' \leq k} \Omega_c^{k'} \neq \emptyset$, i.e., the available servers of group $ k+1 $ overlap with the participating servers of lower-indexed groups.

If $ \max_{m \in \Omega_c^{k+1}} b_m^c(k+1) \leq \textsc{WF}_k $, we again have $ \textsc{WF}_{k+1} = \textsc{WF}_k $ by \eqref{wf_delta} and thus, $ \textsc{WF}_{k+1} - \textsc{WF}_k = 0 < \textsc{OPT} $.

If $ \max_{m \in \Omega_c^{k+1}} b_m^c(k+1) > \textsc{WF}_k $, we have $ \textsc{WF}_{k+1} > \textsc{WF}_k $. In this case, by definition, each participating server of groups $ 1, \ldots, k $ must also be a participating server of group $ k+1 $ by WF (because the estimated busy times of all these servers are smaller than $ \max_{m \in \Omega_c^{k+1}} b_m^c(k+1) $).

For each $ m \in \mathcal{S}_c^{k+1} \backslash \Omega_c^{k+1} $ (i.e., server $ m $ is available to group $ k+1 $ but does not participate in $ \mathcal{T}_c^{k+1} $'s assignment), it can be shown that:
\begin{enumerate}
    \item $ b_m^c(k) \geq \max_{m' \in \Omega_c^{k+1}} b_{m'}^{c}(k+1) $, and
    \item by WF, server $ m $ does not participate in the assignment of groups $ 1, \ldots, k $.
\end{enumerate}
Property 1) follows directly from the definition of WF (otherwise, server $ m $ must be a participating server of group $ k+1 $). Property 2) follows from the above observation that each participating server of groups $ 1, \ldots, k $ must also be a participating server of group $ k+1 $ by WF. By properties 1) and 2), for each server $ m \in \mathcal{S}_c^{k+1} \backslash \Omega_c^{k+1} $, the initial estimated busy time $b_m^c(0)$ satisfies
\begin{equation}
    b_m^c(0) = b_m^c(k) \geq \max_{m' \in \Omega_c^{k+1}} b_{m'}^c(k+1).
    \label{eq_042502}
\end{equation}

With this result, we prove $ \textsc{WF}_{k+1} - \textsc{WF}_k \leq \textsc{OPT} $ by comparing $ \Omega_c^{k+1} $ and $ \beth_c^{k+1} $, where $ \beth_c^{k+1} \subseteq \mathcal{S}_c^{k+1} $ is the set of participating servers for the assignment of $\mathcal{T}_c^{k+1}$ by $ \textsc{OPT}_{k+1} $. 

\textbf{Case II-A.} $ \beth_c^{k+1} \cap (\mathcal{S}_c^{k+1} \backslash \Omega_c^{k+1}) \neq \emptyset $, i.e., at least one server in $ \mathcal{S}_c^{k+1} \backslash \Omega_c^{k+1} $ is chosen for assignment of $ \mathcal{T}_c^{k+1} $ by $ \textsc{OPT}_{k+1} $. Taking any server $ m \in \beth_c^{k+1} \cap (\mathcal{S}_c^{k+1} \backslash \Omega_c^{k+1}) $, we have
\begin{equation}
    \textsc{WF}_{k+1} = \max_{m' \in \Omega_c^{k+1}} b_{m'}^c (k+1) \overset{\eqref{eq_042502}}{\leq} b_m^c(0) < \textsc{OPT}_{k+1} \overset{\eqref{opt_k}}{\leq} \textsc{OPT}.
\end{equation}
As a result, $ \textsc{WF}_{k+1} - \textsc{WF}_k < \textsc{WF}_{k+1} < \textsc{OPT} $.

\textbf{Case II-B.} $ \beth_c^{k+1} \cap (\mathcal{S}_c^{k+1} \backslash \Omega_c^{k+1}) = \emptyset $. It follows that $ \beth_c^{k+1} \subseteq \Omega_c^{k+1} $, which implies
\begin{equation}
    \beth_c^{k+1} = \Omega_c^{k+1} \cap \beth_c^{k+1}.
    \label{eq_042508}
\end{equation}
We divide the servers in $ \Omega_c^{k+1} $ into two disjoint sets:
\begin{align}
    \Omega_{c, \leq}^{k+1} := \left\{ m \in \Omega_c^{k+1} \mid b_m^c(k) \leq \textsc{WF}_k \right\}, \label{eq_042504} \\
    \Omega_{c, >}^{k+1} := \left\{ m \in \Omega_c^{k+1} \mid b_m^c(k) > \textsc{WF}_k \right\}. \label{eq_042505}
\end{align}
For each $ m \in \Omega_{c, >}^{k+1} $, if server $ m $ is a participating server in the assignment of any group $ 1, \ldots, k $ by WF, we have
\begin{equation}
    \textsc{WF}_k \overset{\eqref{eq_042503}}{\geq} b_m^c(k) \overset{\eqref{eq_042505}}{>} \textsc{WF}_k,
\end{equation}
giving a contradiction. Thus, server $ m $ is not a participating server in the assignment of groups $ 1, \ldots, k $ by WF. As a result, $ b_m^c(k) = b_m^c(0) $ holds for each $ m \in \Omega_{c, >}^{k+1} $.

\textbf{Case II-B1.} $ \beth_c^{k+1} \cap \Omega_{c, >}^{k+1} = \emptyset $. In this case, we have $ \beth_c^{k+1} \subseteq \Omega_{c, \leq}^{k+1} $. In $ \textsc{OPT}_{k+1} $, $ \textsc{OPT}_{k+1} $ time slots of all the participating servers $ \beth_c^{k+1} $ can accommodate all the tasks of $ \mathcal{T}_c^{k+1} $. Thus,
\begin{equation}
    |\mathcal{T}_c^{k+1}| \leq \textsc{OPT}_{k+1} \cdot \sum_{m \in \beth_c^{k+1}} \mu_m^c.
    \label{eq_042506}
\end{equation}
On the other hand, by the definition of WF, $ (\textsc{WF}_{k+1} - \textsc{WF}_k - 1) $ time slots of all the participating servers $ \Omega_{c, \leq}^{k+1} $ are not adequate for all the tasks of $ \mathcal{T}_c^{k+1} $. Hence,
\begin{equation}
    \left( \textsc{WF}_{k+1} - \textsc{WF}_k - 1 \right) \cdot \sum_{m \in \Omega_{c, \leq}^{k+1}} \mu_m^c < |\mathcal{T}_c^{k+1}|.
    \label{eq_042507}
\end{equation}
Combining \eqref{eq_042506} and \eqref{eq_042507} and noting that $ \beth_c^{k+1} \subseteq \Omega_{c, \leq}^{k+1} $, we obtain $ \textsc{WF}_{k+1} - \textsc{WF}_k - 1 < \textsc{OPT}_{k+1} $. By \eqref{opt_k}, we have $ \textsc{WF}_{k+1} - \textsc{WF}_k \leq \textsc{OPT} $.

\textbf{Case II-B2.} $ \beth_c^{k+1} \cap \Omega_{c, >}^{k+1} \neq \emptyset $. In this case, at least one server in $ \Omega_{c, >}^{k+1} $ is a participating server for the assignment of $ \mathcal{T}_c^{k+1} $ by $ \textsc{OPT}_{k+1} $. Hence, all the servers in $ \Omega_{c, \leq}^{k+1} $ must also be participating servers for the assignment of $ \mathcal{T}_c^{k+1} $ by $ \textsc{OPT}_{k+1} $. In other words, $ \Omega_{c, \leq}^{k+1} \subseteq \beth_c^{k+1} $.

\begin{figure}[h]
  \centering
  \includegraphics[width=3.1in]{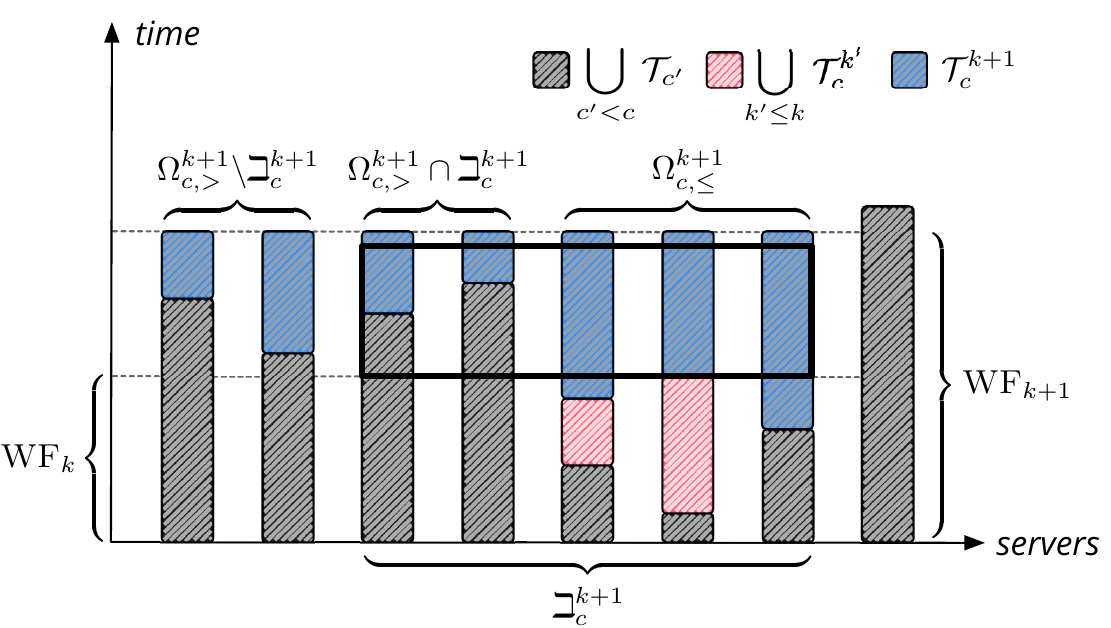}
  \caption{Visualization of Case II-B2.}
  \label{blackbox}
\end{figure}

In $ \textsc{OPT}_{k+1} $, $ \textsc{OPT}_{k+1} $ time slots of all the participating servers $ \beth_c^{k+1} $ can accommodate all the tasks of $ \mathcal{T}_c^{k+1} $ and the outstanding tasks of previous jobs queued at the servers in $ \beth_c^{k+1} \cap \Omega_{c, >}^{k+1} $. Thus,
\begin{equation}
    |\mathcal{T}_c^{k+1}| + \sum_{m \in \Omega_{c, >}^{k+1} \cap \beth_c^{k+1}} b_m^c(0) \cdot \mu_m^c \leq \textsc{OPT}_{k+1} \cdot \sum_{m \in \beth_c^{k+1}} \mu_m^c.
    \label{eq_042510}
\end{equation}

On the other hand, by the definition of WF, $ (\textsc{WF}_{k+1} - \textsc{WF}_k - 1) $ time slots of all the participating servers in $ \Omega_c^{k+1} \cap \beth_c^{k+1} $ are not adequate for all the tasks of $ \mathcal{T}_c^{k+1} $ and the backlogs beyond $ \textsc{WF}_k $ (see the black box in Fig. \ref{blackbox} for an illustration). Hence,
\begin{align}
    &(\textsc{WF}_{k+1} - \textsc{WF}_k - 1) \cdot \sum_{m \in \beth_c^{k+1}} \mu_m^c \nonumber\\
    &\qquad\overset{\eqref{eq_042508}}{=} (\textsc{WF}_{k+1} - \textsc{WF}_k - 1) \cdot \sum_{m \in \Omega_c^{k+1} \cap \beth_c^{k+1}} \mu_m^c \nonumber\\ %&& \text{by \eqref{eq_042508}} \nonumber\\
    &\qquad< |\mathcal{T}_c^{k+1}| + \sum_{m \in \Omega_{c, >}^{k+1} \cap \beth_c^{k+1}} (b_m^c(0) - \textsc{WF}_k) \cdot \mu_m^c.
    \label{eq_042509}
\end{align}

Note that the right side of \eqref{eq_042509} is no greater than the left side of \eqref{eq_042510}. Combining \eqref{eq_042510} with \eqref{eq_042509}, we obtain $ \textsc{WF}_{k+1} - \textsc{WF}_k - 1 < \textsc{OPT}_{k+1} $. By \eqref{opt_k}, we have $ \textsc{WF}_{k+1} - \textsc{WF}_k \leq \textsc{OPT} $.

With the above analysis for \textbf{Case I} and \textbf{Case II}, we have shown that \eqref{delta} indeed holds.

\bibliographystyle{IEEEtran}
\bibliography{ref}

\end{document}